\documentclass[conference, 10pt]{IEEEtran}
\IEEEoverridecommandlockouts   

\def\BibTeX{{\rm B\kern-.05em{\sc i\kern-.025em b}\kern-.08em
    T\kern-.1667em\lower.7ex\hbox{E}\kern-.125emX}}

\usepackage{supertech-llncs-light}

\newcommand{\algName}{PACO}
\newcommand{\ouralg}{\algName{} algorithm}
\newcommand{\ouralgs}{\algName{} algorithms}
\newcommand{\ourAlg}[1]{\algName{} \proc{#1} algorithm}
\newcommand{\ourAlgS}[1]{\algName{} \proc{#1}'s algorithm}
\newcommand{\ourModel}{ideal distributed cache model}
\newcommand{\propName}{Perfect Strong Scaling}
\newcommand{\ourprop}{\propName{} Property}
\newcommand{\perfectspeedup}{perfect strong scaling}
\newcommand{\Obc}{Optimal balanced computation}
\newcommand{\obc}{optimal balanced computation}
\newcommand{\Occ}{Optimal balanced communication}
\newcommand{\occ}{optimal balanced communication}
\newcommand{\qsum}{Q^{\sum}_p}
\newcommand{\tsum}{T^{\sum}_p}
\newcommand{\qinf}{Q^{\max}_p}
\newcommand{\tinf}{T^{\max}_p}
\newcommand{\omg}{\omega_0}
\newcommand{\sr}{\gamma}
\renewcommand{\emph}[1]{\textbf{\textit{#1}}}

\begin{document}

\title{Balanced Partitioning of Several Cache-Oblivious Algorithms
\thanks{This research is supported in part by Shanghai Natural
Science Funding (No. 18ZR1403100), and National Science Foundation
of China under Grant No. 11690013, 71991471, U1811461}
}


\author{\IEEEauthorblockN{Yuan Tang
}
\IEEEauthorblockA{School of Computer Science, School of Software \\
Fudan University\\
Shanghai, P. R. China \\
yuantang@fudan.edu.cn}
\and
\IEEEauthorblockN{Weiguo Gao}
\IEEEauthorblockA{School of Mathematical Sciences, School of Data Science \\
Fudan University\\
Shanghai, P. R. China \\
wggao@fudan.edu.cn}
}

\maketitle

\begin{abstract}

Frigo et al. proposed an ideal cache model and a recursive 
technique to design sequential cache-efficient algorithms in a 
cache-oblivious fashion. 
Ballard et al. pointed out that it is a fundamental open
problem to extend the technique to an arbitrary architecture. 
Ballard et al. raised another open question on how to
parallelize Strassen's algorithm exactly and efficiently on an
arbitrary number of processors.

We propose a novel way of partitioning a
cache-oblivious algorithm to achieve perfect strong scaling
on an arbitrary number, even a prime number, of 
processors within a certain range in a shared-memory setting.
Our approach is Processor-Aware but Cache-Oblivious (PACO).
We demonstrate our approach on several important cache-oblivious
algorithms, including LCS, 1D, GAP, classic rectangular matrix
multiplication on a semiring, and Strassen's algorithm.
%
We discuss how to extend our approach to a 
distributed-memory architecture, or even a heterogeneous 
computing system. Hence, our work may provide a new perspective 
on the fundamental open problem of
extending the recursive cache-oblivious technique to an arbitrary 
architecture.
%
We provide an almost exact solution to
the open problem on parallelizing Strassen.
\punt{
Our approach may provide a new perspective on extending the 
recursive cache-oblivious technique to an arbitrary architecture. 
}
All our algorithms demonstrate better scalability or better 
overall parallel cache complexities than the best known algorithms.
%
%
Preliminary experiments justify our theoretical prediction that
the PACO algorithms can outperform 
significantly state-of-the-art Processor-Oblivious (PO) and 
Processor-Aware (PA) counterparts.
%
\end{abstract}

\punt{
}


\secput{intro}{Introduction}
%
%
\begin{table*}[!ht]
\centering
\begin{tabular}{ccc}
    \toprule
    Algo. & Time ($T_p$ or $\tinf$) & Overall Parallel Cache ($Q_p$ or $\qsum$) \\
    \midrule %
    PO LCS \cite{FrigoSt09, ColeRa12a} & $O(n^2/p + n^{\log_2 3})$ & $O(n^2/(LZ) + \sqrt{p n^{3.58}} + p n^{1.58})$ \\ 
    PA LCS \cite{ChowdhuryRa08} & $2 n^2/p + o(n^2/p)$ & $O(n^2/(LZ) + pn/L)$ \\
    \emph{\algName} LCS (\thmref{lcs}) & $n^2/p + o(n^2/p)$ & $O(\min\{n^2/(LZ) + (p n \log (pZ)) / L, (p n \log n) / L\})$ \\
    \midrule
    PO 1D \cite{BlellochGu18} & $O(n^2/p + n)$ & $O(n^2/(LZ) + (p n Z) / L)$ \\ 
    sublinear 1D \cite{GalilPa94} & $O(n^2/p + \sqrt{n} \log n)$ & $O(n^2/L + (p (\sqrt{n} \log n) Z) / L)$ \\
    \emph{\algName} 1D (\thmref{1D}) & $O(n^2/p)$ & $O(\min \{n^2/(LZ) + (p Z \log Z) / L, (\sqrt{p} n \log n) / L\})$ \\
    \midrule
    PO GAP \cite{BlellochGu18} & $O(n^3/p + n^{\log_2 3})$ & $O(\frac{n^3}{LZ} + (n^2 \cdot \min \{\frac{\log_2 n}{\sqrt{Z}}, \log_2 \sqrt{Z}\}) / L + \frac{p n^{\log_2 3} Z}{L})$ \\
    sublinear GAP \cite{GalilPa94} & $O(n^4/p + \sqrt{n} \log n)$ & $O(n^4/L + (p (\sqrt{n} \log n) Z) / L)$ \\
    \emph{\algName} GAP (\thmref{gap}) & $O(n^3/p)$ & $O(\min\{n^3/(LZ) + (n^2 \log Z) / L, (n^2 \log n) / L\})$ \\ 
    \midrule
    PO MM \cite{ColeRa12} & $O(n^3/p + \log^2 n)$ & $Q_1 + O((p \log p)^{1/3} \cdot n^2/L + p \log p)$ \\
    PA MM \cite{DemmelElFo13} & \multicolumn{2}{c}{same bounds as \algName{}, except $p$ can not have large prime factors}\\
    \emph{\algName} MM & $O(nmk/p + n + m + k)$ & $Q_1 + 
    O(\min\{pmk, \allowbreak \sqrt{pnmk^2}, \allowbreak p^{1/3}(nmk)^{2/3}\}/L)$ \\
    \midrule
    PO Strassen \cite{ColeRa12} & $O(n^{\omg}/p + \log^2 n)$ & $O(n^{\omg}/(LZ^{\omg / 2 - 1}) + (p \log p)^{1/3} \cdot n^2/L + p \log p)$ \\
    PA Strassen \cite{BallardDeHo12, LipshitzBaDe12} & \multicolumn{2}{c}{same bounds as \algName{}, except $p = m \cdot 7^k$, where $1 \leq m < 7$ and $1 \leq k$ are integers} \\
    \emph{\algName} Strassen (\thmref{strassen}) & $O(n^{\omg}/p)$ & $O(n^{\omg}/(LZ^{\omg / 2 - 1}) + n^2/(Lp^{2/\omg - 1}))$ \\
    \midrule
    deterministic PO Sorting \cite{ColeRa17, ColeRa12} & $O((n/p) \log n + \log n \log \log n)$ & $O((n/L) \log_{Z} n + p \frac{\log n}{\log (n/p)} \cdot L)$ \\
    \emph{\algName} Sort (\thmref{sort}) & $O((1 + \epsilon) (n / p) \cdot \log n)$ & $O((n/L) \log_{Z} (n/p))$ \\
    \bottomrule
\end{tabular}%
\caption{Main results of this paper, comparing with
    typical prior works.
    ``PO'' : processor-oblivious; ``PA'' :
    processor-aware; ``\algName'' : processor-aware but
    cache-oblivious.
    ``$Q_1$'' : optimal sequential cache complexity;
    ``$T_p, Q_p$'' : notations for PO algorithms;
    ``$\tinf, \qsum$'' : notations for PA and PACO algorithms;
    ``$Z$'' : cache size; ``$L$'' : cache line size;
    ``$\epsilon$'' : small constant;
    ``$\omg = \log_2 7$'';
    }
\label{tab:contri-table}
\end{table*}%


\begin{table}
\centering
\caption{Acronyms \& Notations}
\label{tab:symbols}
\begin{tabular}{cc}
\toprule
\multicolumn{2}{c}{General Acronyms}\\
\midrule
LCS & Longest Common Subsequence\\
MM  & Matrix Multiplication\\  
RWS & Randomized Work-Stealing\\
w.h.p. & with high probability\\
PO  & Processor-Oblivious \\
PA  & Processor-Aware \\
PACO & Processor-Aware Cache-Oblivious\\
\midrule
\multicolumn{2}{c}{Parameters}\\
\midrule
$n, m, k$ & Input sizes\\
$\epsilon_i$ & small constant\\
$p$ & Number of processors\\
$Z$ & cache size\\
$L$ & cache line size\\
\midrule
\multicolumn{2}{c}{Complexity Notations for PO Alg.}\\
\midrule
$T_1$ & total work\\ 
$T_\infty$ & work along critical path (time, span, depth)\\
$T_p$ & parallel running time on $p$ processors\\
$Q_1$ & sequential cache complexity\\
$Q_p$ & overall cache complexity over $p$ processors\\
\midrule
\multicolumn{2}{c}{Complexity Notations for PACO Alg.}\\
\midrule
$\tsum$ & overall work over $p$ processors\\
$\tinf$ & work along critical path\\
$\qsum$ & overall cache complexity over $p$ processors\\
$\qinf$ & cache complexity along critical path\\
\bottomrule
\end{tabular}
\vspace{-0.4cm}
\end{table}

\newcommand{\len}{1}
\tikzset{
    bgap/.pic={
\coordinate (O) at (0, 0, 0);
\coordinate (A) at (\len, 0, 0);
\coordinate (B) at (\len, 0, \len);
\coordinate (C) at (0, 0, \len);
\coordinate (D) at (0, \len, 0);
\coordinate (E) at (\len, \len, 0);
\coordinate (F) at (\len, \len, \len);
\coordinate (G) at (0, \len, \len);
\coordinate (H) at (0, 2 * \len, 0);

\draw (C) -- (G) -- (B) -- cycle;
\draw (E) -- (A) -- (B) -- cycle;
\draw (G) -- (E) -- (B) -- cycle;
\draw (H) -- (G) -- (E) -- cycle;
\draw (H) -- (B);
\draw [dotted] (D) -- (O) -- (C);
\draw [dotted] (O) -- (A);
\draw [dotted] (H) -- (D) -- (G);
\draw [dotted] (D) -- (E);
}}
\tikzset{
    lgap/.pic={
\coordinate (O) at (0, 0, 0);
\coordinate (A) at (\len, 0, 0);
\coordinate (B) at (\len, 0, \len);
\coordinate (C) at (0, 0, \len);
\coordinate (D) at (0, \len, 0);
\coordinate (E) at (\len, \len, 0);
\coordinate (F) at (\len, \len, \len);
\coordinate (G) at (0, \len, \len);
\coordinate (H) at (0, 2 * \len, 0);

\draw (C) -- (G) -- (B) -- cycle;
\draw [dotted] (E) -- (A) -- (B) -- cycle;
\draw (G) -- (E) -- (B) -- cycle;
\draw (H) -- (G) -- (E) -- cycle;
\draw (H) -- (B);
\draw [dotted] (D) -- (O) -- (C);
\draw [dotted] (O) -- (A);
\draw [dotted] (H) -- (D) -- (G);
\draw [dotted] (D) -- (E);
}}
\tikzset{
    rgap/.pic={
\coordinate (O) at (0, 0, 0);
\coordinate (A) at (\len, 0, 0);
\coordinate (B) at (\len, 0, \len);
\coordinate (C) at (0, 0, \len);
\coordinate (D) at (0, \len, 0);
\coordinate (E) at (\len, \len, 0);
\coordinate (F) at (\len, \len, \len);
\coordinate (G) at (0, \len, \len);
\coordinate (H) at (0, 2 * \len, 0);

\draw (H) -- (G) -- (B) -- (E) -- cycle;
\draw (H) -- (B);
\draw (G) -- (E) -- (A);
\draw (A) -- (B);
\draw [dotted] (B) -- (C);
\draw [dotted] (D) -- (E);
\draw [dotted] (O) -- (A);
\draw [dotted] (C) -- (O) -- (D) -- (G);
\draw [dotted] (D) -- (H);
}}
\tikzset{
    tgap/.pic={
\coordinate (O) at (0, 0, 0);
\coordinate (A) at (\len, 0, 0);
\coordinate (B) at (\len, 0, \len);
\coordinate (C) at (0, 0, \len);
\coordinate (D) at (0, \len, 0);
\coordinate (E) at (\len, \len, 0);
\coordinate (F) at (\len, \len, \len);
\coordinate (G) at (0, \len, \len);
\coordinate (H) at (0, 2 * \len, 0);

\draw (H) -- (G) -- (B) -- (E) -- cycle;
\draw (G) -- (E);
\draw (H) -- (B);
\draw [dotted] (O) -- (C) -- (G) -- (D) -- cycle;
\draw [dotted] (D) -- (E) -- (A) -- (O);
\draw [dotted] (H) -- (O);
\draw [dotted] (C) -- (B) -- (A);
}}
\tikzset{
    fbox/.pic={
\coordinate (O) at (0, 0, 0);
\coordinate (A) at (\len, 0, 0);
\coordinate (B) at (\len, 0, \len);
\coordinate (C) at (0, 0, \len);
\coordinate (D) at (0, \len, 0);
\coordinate (E) at (\len, \len, 0);
\coordinate (F) at (\len, \len, \len);
\coordinate (G) at (0, \len, \len);

\draw (C) -- (B) -- (F) -- (G) -- cycle;
\draw (F) -- (E) -- (A) -- (B);
\draw (G) -- (D) -- (E);
\draw [dotted] (D) -- (O) -- (C);
\draw [dotted] (O) -- (A);
}}
\tikzset{
    bbox/.pic={
\coordinate (O) at (0, 0, 0);
\coordinate (A) at (\len, 0, 0);
\coordinate (B) at (\len, 0, \len);
\coordinate (C) at (0, 0, \len);
\coordinate (D) at (0, \len, 0);
\coordinate (E) at (\len, \len, 0);
\coordinate (F) at (\len, \len, \len);
\coordinate (G) at (0, \len, \len);

\draw [dotted] (C) -- (B) -- (F) -- (G) -- cycle;
\draw [dotted] (F) -- (E) -- (A) -- (B);
\draw [dotted] (G) -- (D) -- (O) -- (C);
\draw [dotted] (D) -- (E);
\draw [dotted] (O) -- (A);
}}
\tikzset{
    lbox/.pic={
\coordinate (O) at (0, 0, 0);
\coordinate (A) at (\len, 0, 0);
\coordinate (B) at (\len, 0, \len);
\coordinate (C) at (0, 0, \len);
\coordinate (D) at (0, \len, 0);
\coordinate (E) at (\len, \len, 0);
\coordinate (F) at (\len, \len, \len);
\coordinate (G) at (0, \len, \len);

\draw (C) -- (B) -- (F) -- (G) -- cycle;
\draw [dotted] (F) -- (E) -- (A) -- (B);
\draw [dotted] (G) -- (D) -- (O) -- (C);
\draw [dotted] (D) -- (E);
\draw [dotted] (O) -- (A);
}}
\tikzset{
    rbox/.pic={
\coordinate (O) at (0, 0, 0);
\coordinate (A) at (\len, 0, 0);
\coordinate (B) at (\len, 0, \len);
\coordinate (C) at (0, 0, \len);
\coordinate (D) at (0, \len, 0);
\coordinate (E) at (\len, \len, 0);
\coordinate (F) at (\len, \len, \len);
\coordinate (G) at (0, \len, \len);

\draw (F) -- (E) -- (A) -- (B) -- cycle;
\draw [dotted] (D) -- (E);
\draw [dotted] (O) -- (A);
\draw [dotted] (B) -- (C);
\draw [dotted] (F) -- (G);
\draw [dotted] (O) -- (C);
\draw [dotted] (O) -- (D);
}}

Frigo et al. proposed an ideal cache model and a recursive
technique to design sequential cache-efficient 
algorithm on a hierarchical architecture of caches in a
cache-oblivious fashion \cite{FrigoLePr12}. 
That is, an algorithm does not specify any 
parameters on cache architecture such as number of cache levels, 
cache size of each level, or block transfer sizes between each 
pair of consecutive levels, but still can attain asymptotically 
optimal cache complexity on all levels of cache.
Ballard et al. (Sect. $6.2$ of \cite{BallardDeHo13}) pointed out 
that it is a fundamental open problem to extend the technique to 
an arbitrary architecture.

In the literature, there are two classes of extension. One is 
processor-oblivious (PO) and the other is processor-aware (PA).
A PO approach does not use the knowledge of processor number,
cache architecture \cite{ColeRa17, Chowdhury07, BlellochChGi08, BlellochGiSi10, BlellochFiGi11, ChowdhuryRaSi13, ColeRa11, ChowdhurySiBl10, BilardiPiPu16}, or network architecture \cite{BilardiPiPu16}.
An algorithm just exploits maximal parallelism and 
bounds its sequential cache complexity to be optimal, then 
relies on a runtime scheduler or folding
mechanism \cite{BilardiPiPu16} to yield a provably efficient 
solution
on either a shared-memory or a distributed-memory architecture.
The main benefits of the PO approach are easy-of-programming,
simple and scalable to
an arbitrary number of processors within a certain range.
However, 
Frigo and Strumpen \cite{FrigoSt09} pointed out that the 
communication complexity (cache miss in a shared-memory setting,
and message bandwidth and latency in a distributed-memory setting)
of a PO algorithm may not be as good as its PA counterpart.
Though Blelloch et al. \cite{BlellochGiSi10} show that if a PO
algorithm has a poly-logarithmic depth, i.e. low-depth, it will 
have low cache complexity on a shared-memory architecture, we can 
see from \tabref{contri-table} that a PA counterpart can still be 
better.
\punt{
For example, though the deterministic PO (actually 
resource-oblivious) sorting \cite{ColeRa17}
has only an $O(\log n \log \log n)$ depth, 
where $n$ is the problem size, it incurs at least as many
cache misses as the best sequential algorithm, while 
our \ourAlg{Sort}{} can incur less.
}

On the other hand, a PA approach \cite{ChowdhuryRa08, BallardDeHo12, SolomonikDe11, BallardDeHo11, DemmelElFo13} utilizes the 
knowledge of processor number, sometimes even the knowledge of
cache / memory architecture,
to provide a strong scaling \cite{BallardDeHo12a, BallardDeHo12b} 
algorithm in terms of both computation
and communication. However, classic PA algorithms may not
utilize all processors effectively unless
the processor number matches well the structure of algorithm.
For example, a straightforward implementation of the
Communication-Avoiding Parallel Strassen (CAPS) algorithm by
Ballard et al. \cite{BallardDeHo12} requires processor number
$p$ to be an exact power of $7$.
\punt{
Though an implementation can
simply ignore no more than $6/7$ of $p$ and incurs
only a constant factor more overheads. In practice, constant
factors matter, especially for such a significant fraction.
}
Lipshitz et al.  \cite{LipshitzBaDe12} later improved the 
required processor number to a multiple of $7$ with no large prime
factors, i.e. $p = m \cdot 7^x$, where $1 \leq m < 7$ and
$1 \leq x$ are integer numbers, by a hybrid of Strassen and 
classic matrix multiplication (MM) algorithms.
This hybrid algorithm may still lose up to $1/7$ of
the available processors, which can nullify the performance
advantage of Strassen in practice.
So Ballard et al. (Sect. $6.5$ of \cite{BallardDeHo12}) raised
an open question if
it is possible to run Strassen's algorithm concurrently on an
arbitrary number of processors, while still attaining the
computation lower bound exactly and attaining the communication
lower bound up to a constant factor.

\punt{
For example, the PA LCS algorithm by Chowdhury and
Ramachandran \cite{ChowdhuryRa08} makes a $p$-way 
divide-and-conquer at the top level of recursion then switches
to the sequential $2$-way divide-and-conquer afterwards. 
This algorithm does not fully utilize all processors all the time,
hence its critical-path length is a factor of $2$ larger than
ours (\tabref{contri-table}). In practice, constant matters as 
we can see from the experimental results in \secref{expr}.
}

\paragrf{Contributions (\tabref{contri-table}): }
We propose a novel way of partitioning a 
cache-oblivious algorithm to achieve perfect strong scaling
in a shared-memory setting based on a pruned BFS
traversal of the algorithm's divide-and-conquer tree.
Our approach uses processor number $p$, but no knowledge on
cache architecture, hence is Processor-Aware but
Cache-Oblivious (PACO). Our PACO approach does not assume any
special property of $p$, e.g. factorizable into two or three
roughly equal numbers or does not contain large prime factors,
etc. so that it works for an arbitrary $p$ within a certain range.
We demonstrate our approach on several important 
cache-oblivious algorithms, including longest common
subsequence (LCS) (\secref{lcs}), which is Dynamic Programming
(DP) with constant dependency, 1D problem (\secref{1D}) and
GAP problem (\secref{gap}), both of which are DP with 
more-than-constant dependencies, classic rectangular Matrix 
Multiplication on a semiring and Strassen's algorithm 
(\secref{strassen}), as well as comparison based sorting.
In particular, our \ourAlgS{Strassen}{} attains both the
computation and communication lower bounds on an arbitrary
number of processors in a shared-memory setting. If translated
to a distributed-memory setting, our
\ourAlg{Strassen-Const-Pieces}{} attains the computation
lower bound up to an arbitrarily small constant factor,
attains the bandwidth lower bound up to a constant factor,
and has an $O(\log p)$ latency bound.
We also conjecture that this latency bound is tight up
to a constant factor.
Hence we provide an almost exact solution to the open
problem \cite{BallardDeHo12} on parallelizing Strassen.

Compared with classic PA approaches, our algorithms achieve 
perfect strong scaling on an arbitrary number, even a prime 
number, of processors within a certain range. So we argue that 
our approach is as scalable as classic PO approaches.
We discuss how to possibly extend our approach to a
distributed-memory setting, or even a heterogeneous
computing system. Hence, our work may provide a new
perspective on the fundamental open problem \cite{BallardDeHo13}
on extending the recursive cache-oblivious \cite{FrigoLePr12}
technique to an arbitrary architecture.
Compared with classic PO approaches, our algorithms usually
attain a better communication complexity, no matter 
the best PO counterpart has a poly-logarithmic (low-depth) 
\cite{BlellochGiSi10} or super-linear \cite{CormenLeRi09, ChowdhuryRa06, Leiserson06}
critical-path length.
Our work may not only initiate new ideas on 
designing provably efficient runtime scheduler, but also
provide a new perspective on the fundamental
open problem of extending a sequential cache-oblivious algorithm
to an arbitrary architecture.
%
Preliminary experiments show that 
our new algorithms outperform state-of-the-art
PO and PA counterparts significantly in practice.

\secput{model}{Models}

%
We view a parallel computation as a Directed Acyclic 
Graph (DAG). Each vertex stands for a piece of computation 
with no parallel construct and each directed edge
represents some data dependency between a pair
of vertices.
For simplicity, we count each arithmetic operation such as
multiplication, addition, and comparison uniformly as an
$O(1)$ operation.
Our DAG considers only data dependency because any extra 
control dependency is artificial dependency, which can only 
hurt potential parallelism \cite{TangYoKa15, DinhSiTa16}.
Our computation DAG is slightly different from the CDAG of
\cite{BallardDeHo13}. In a CDAG, each vertex stands for
an input / intermediate / output argument, and each edge stands 
for a direct dependency. The difference is due to that CDAG 
counts the number of edges to bound communication cost, while 
our DAG calculates the computation and communication complexity 
of each task, i.e. a subset of vertices and edges of DAG,
independently. 
%

\begin{figure}
\centering
\begin{tikzpicture}[scale=1, transform shape]
\path (0, 2) node [draw, circle, thick] (p1) {$P_1$}
    -- (0, 1) node [draw, thick] (c1) {$M$};
\path (1, 2) node [draw, circle, thick] (p2) {$P_2$}
    -- (1, 1) node [draw, thick] (c2) {$M$};
\path (3, 2) node [draw, circle, thick] (pp) {$P_p$}
    -- (3, 1) node [draw, thick] (cp) {$M$};
\draw [thick] (p1) -- (c1) -- (0, 0) node[pos=.5, left]{$B$};
\draw [thick] (p2) -- (c2) -- (1, 0) node[pos=.5, left]{$B$};
\draw [thick] (pp) -- (cp) -- (3, 0) node[pos=.5, left]{$B$};
\draw[very thick] (-1, 0) -- (4, 0) node[pos=.9, above]{$\infty$};

\foreach \h in {1}
\foreach \s in {.4}
\foreach \r in {0.05}{
\fill[black] (1.6, \h) circle (\r) 
    -- (1.6+\s, \h) circle (\r)
    -- (1.6+2*\s, \h) circle (\r);
};
\end{tikzpicture}
\caption{\ourModel}
\label{fig:p-distributed-ideal-cache-model}
\end{figure}

We adopt the \ourModel{}
(\figref{p-distributed-ideal-cache-model}) proposed
by Frigo and Strumpen \cite{FrigoSt09} as our machine model.
It is a two-level memory model. There are $p$ 
dedicated processors with identical computing power, each of 
which is equipped with a private ideal cache. An ideal cache is
fully associative and is managed by an omniscient, i.e. off-line 
optimal, cache replacement policy that replaces the cache line
whose next access is the farthest in future \cite{Belady66}. 
The papers of 
\cite{FrigoLePr12, SleatorTa85} justify the ideal cache assumption.
Each private cache is of size $Z$. All caches are connected by an 
arbitrarily large shared memory. Private caches exchange data
with shared memory atomically in cache line of size $L$.
A processor can only access data in its private cache. If a value
is not present in the cache, the processor incurs a cache miss to
bring the data from shared memory to its cache.
We do not consider cache-coherence protocols because all 
algorithms considered in this paper do not have data race
\footnote{Data race means that there are at least two 
processors accessing the same location of shared memory
simultaneously, at least one of which are ``write''.}, nor do we
consider false sharing.
All private caches are \emph{non-interfering},
i.e. the number of cache misses incurred by one processor can be
analyzed independently of the actions of other processors in the
system. This assumption is valid under the DAG-consistent memory
model maintained by the Backer protocol \cite{BlumofeFrJo96} or 
the HSMS model \cite{AcarBlBl00}.

\secput{alg-principle}{\ouralg{} Design and Analysis}

\paragrf{General \ouralg{}: }
Based on the observation that the maximal speedup 
attainable on a $p$-processor system is usually $p$-fold so
that excessive parallelism may not be necessary, we have a 
general \ouralg{} as follows.
The algorithm traverses a $c$-way divide-and-conquer tree of a
cache-oblivious algorithm in a pruned breadth-first (BFS) fashion,
where $c$ is a small algorithm-specific constant.
That is, it unfolds the tree depth by depth in a breadth-first
fashion. As soon as it figures out that some depth
has equal or more than $p$ nodes that have all inputs ready
and have no data dependency among each other, it cuts off 
(prunes) up to $(c-1) \cdot p$ of them and assigns to $p$ 
processors in a round-robin fashion.
The rest of nodes will stay in the tree and go to more rounds
of ``\defn{pruned BFS}'' traversal.
This procedure repeats until either all nodes are pruned,
i.e. assigned to processors, or all nodes are of base-case 
(constant) size, in 
which case they will be assigned to all processors in a 
round-robin fashion.
\figref{pruned-bfs-pic} shows a pruned BFS traversal of a binary
tree, assuming $p = 3$. Labels indicate assigned (pruned) order.

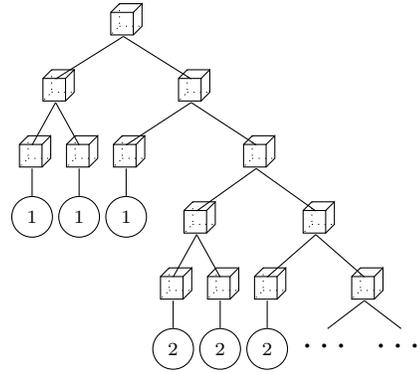
\begin{figure}
\centering
\begin{tikzpicture} [scale=1, transform shape ]
\tikzset{level distance=25pt}
\renewcommand{\len}{1}
\tikzset{
cube/.pic={
\path (0, 0, 0) coordinate (O) -- (\len, 0, 0) coordinate (A) 
    -- (\len, 0, \len) coordinate (B) -- (0, 0, \len) coordinate (C)
    -- (0, \len, 0) coordinate (D) -- (\len, \len, 0) coordinate (E) 
    -- (\len, \len, \len) coordinate (F) -- (0, \len, \len) coordinate (G);

\draw (C) -- (B) -- (F) -- (G) -- cycle;
\draw (F) -- (E) -- (A) -- (B);
\draw (G) -- (D) -- (E);
\draw [dotted] (D) -- (O) -- (C);
\draw [dotted] (O) -- (A);
} }

\Tree 
[.\node[pic-tree node](r){}; 
    [.\node[pic-tree node](r1){}; 
        [.\node[pic-tree node](r1-1){}; \node[id-tree node]{$1$}; ] 
        [.\node[pic-tree node](r1-2){}; \node[id-tree node]{$1$}; ] 
    ]
    [.\node[pic-tree node](r-2){}; 
        [.\node[pic-tree node](r-2-1){}; \node[id-tree node]{$1$}; ] 
        [.\node[pic-tree node](r-2-2){};
            [.\node[pic-tree node](r-2-2-1){}; 
                [.\node[pic-tree node](r-2-2-1-1){}; \node[id-tree node]{$2$}; ] 
                [.\node[pic-tree node](r-2-2-1-2){}; \node[id-tree node]{$2$}; ] 
            ] 
            [.\node[pic-tree node](r-2-2-2){}; 
                [.\node[pic-tree node](r-2-2-2-1){}; \node[id-tree node]{$2$}; ] 
                [.\node[pic-tree node](r-2-2-2-2){}; 
                    [.{\Large $\cdots$} ]
                    [.{\Large $\cdots$} ]
                ]
            ] 
        ]
    ] 
]
\tikzset{cube pic/.style={transform shape, scale=#1, xshift=-5pt, yshift=-2pt}}
\path (r) pic [cube pic=0.3] {cube}
    -- (r1) pic [cube pic=0.3] {cube}
    -- (r1-1) pic [cube pic=0.3] {cube}
    -- (r1-2) pic [cube pic=0.3] {cube}
    -- (r-2) pic [cube pic=0.3] {cube}
    -- (r-2-1) pic [cube pic=0.3] {cube}
    -- (r-2-2) pic [cube pic=0.3] {cube}
    -- (r-2-2-1) pic [cube pic=0.3] {cube}
    -- (r-2-2-1-1) pic [cube pic=0.3] {cube}
    -- (r-2-2-1-2) pic [cube pic=0.3] {cube}
    -- (r-2-2-2) pic [cube pic=0.3] {cube}
    -- (r-2-2-2-1) pic [cube pic=0.3] {cube}
    -- (r-2-2-2-2) pic [cube pic=0.3] {cube};

\end{tikzpicture}
\caption{Pruned BFS Traversal of a Binary Tree ($p = 3$); Labels indicate assigned (pruned) order.}
\label{fig:pruned-bfs-pic}
\end{figure}

\paragrf{Invariant: }
Assuming that each parent node is at least a constant factor 
larger than each of its child node in terms of computation
and communication overheads (volume, surface area, or perimeter
in geometry), we can see that the set of nodes assigned 
to each processor forms an (almost) geometrically decreasing 
sequence and that the top-level node(s) dominate.

\paragrf{Comparison with classic PO approaches: }
Classic PO approach usually recursively divides each and 
every node to base-case size to increase the ``slackness'' 
of an algorithm so that it has better processor utilization 
for a wider range of processor counts.
This more slackness increases the potential deviations 
from its sequential execution order \cite{AcarBlBl00, SpoonhowerBlGi09},
hence usually incurs more communication and synchronization 
overheads than a PA counterpart \cite{FrigoSt09}.

\paragrf{Comparison with classic PA approaches: }
Classic PA approach, on the other hand, may not fully
utilize all processors from beginning to end unless
the processor number matches well the structure of algorithm.
For example, the CAPS algorithm for Strassen  
\cite{BallardDeHo12, LipshitzBaDe12} requires processor number
$p$ to be an exact power of $7$ or at least be a multiple of
$7$ with no large prime factors.

\subsecput{alg-counting}{Complexity counting: }
We count the complexity bounds of a \ouralg{} as follows.
If there is an independent partitioning procedure ahead of
real execution as in the case of LCS (\secref{lcs}), we will 
count them separately.
We assume that any processor starts a task, i.e. a set of
nodes of the divide-and-conquer tree, with
an empty cache and flushes all data to lower-level memory 
when task finishes.
We use notation $\qsum$ to denote all data movements
(cache misses) between upper-level private caches and lower-level 
shared memory summed up over all $p$ processors in cache line 
of size $L$, and notation $\qinf$ to denote the
maximal data movements on any single processor, or along 
a critical path.
Similarly, we have notations $\tsum$ and $\tinf$ for
the amount of computation summed up over all $p$ processors and
the maximal on any single processor, or along a critical path,
respectively.

\paragrf{\ourprop{}: }
We give out a more formal and more strict definition of  
``\ourprop{}'', which was initiated by Ballard et al. 
\cite{BallardDeHo12a, BallardDeHo12b}, as follows.
 
\begin{enumerate}
    \item \defn{\Obc: }
        Firstly, the overall computation complexity ($\tsum$)
        should be asymptotically optimal or match that of
        the best sequential algorithm of the same problem.
        Secondly, the computation complexity on any
        single processor ($\tinf$) should be $O((1/p) \tsum)$.

        We make one more restriction that the amount
        of computation assigned to different processors can 
        differ by no 
        more than an asymptotically smaller term, rather than
        a larger-than-$1$ multiplicative factor.
        
        By the restriction, we make any imbalance of workloads, 
        if any, among different processors
        diminishing when increasing problem size. 
    \item \defn{\Occ: }
        Firstly, the cache misses summed up over all
        processors ($\qsum$) throughout execution should be 
        asymptotically optimal or match that of the best 
        sequential algorithm of the same problem.
        Secondly, the maximal cache misses on any single 
        processor ($\qinf$) should be $O((1/p) \qsum)$.
\end{enumerate}

\paragrf{Discussions: }
The initial notion of ``\perfectspeedup{}'' in 
\cite{BallardDeHo12a, BallardDeHo12b} requires that an algorithm
attains running time on $p$ processors which is linear in $1/p$,
including all communication costs.
Our definition is more formal and more strict in three senses: 
Firstly, it requires that the
overall computation and communication overheads of
algorithm be asymptotically optimal or match that of the best 
sequential algorithm of the same problem; Secondly, it requires 
that any load imbalance among different processors, if any, can 
not be more than an asymptotically smaller term, rather than a 
larger-than-$1$ multiplicative
factor; Thirdly, we require that the property be valid for 
an arbitrary number, even a prime number, of processors within a
certain range.
For example, Galil and Park \cite{GalilPa94} designed a sublinear 
$\tinf = O(\sqrt{n} \log n)$ time (critical-path length) and 
$\tsum = O(n^4)$ overall work algorithm for the GAP 
problem \cite{GalilGi89}.
Due to the sublinear time, the algorithm is \perfectspeedup{} 
according to \cite{BallardDeHo12a, BallardDeHo12b}. Due to the 
asymptotically more work than the optimal
\cite{Chowdhury07, ChowdhuryRa06}, it is not according
to our definition.

For computation and communication overheads, we count both 
overall and along a critical path 
to compare with both PO and PA counterparts.
By convention, a PO algorithm
usually counts its sequential communication complexity and 
critical-path length, then relies on a runtime
scheduler, e.g. \cite{AcarBlBl00, BlumofeLe99}, or a folding 
mechanism, e.g. \cite{BilardiPiPu16}, to yield an overall
parallel computation and communication complexity; On the 
other hand, a PA algorithm, e.g.
\cite{BallardDeHo12, LipshitzBaDe12, SolomonikDe11, BallardDeHo11, DemmelElFo13}, usually calculates overheads along a critical path.

\subsecput{lcs}{\ourAlg{LCS}}
Given two sequences $S = \langle s_1, s_2, \cdots, s_m \rangle$
and $T = \langle t_1, t_2, \cdots, t_n \rangle$, the LCS
problem asks to compute the length of longest common subsequence
\footnote{The subsequence does not have to be contiguous in the
input sequences.} of the two inputs
by the recurrence of \eqref{lcs} \cite{CormenLeRi09}
\footnote{A similar recurrence
applies to the ``pairwise sequence alignment with affine gap
cost'' problem \cite{Gotoh82}. A more complicated case may
further ask to compute the subsequence besides
the length.}:
%
\begin{align}
    & X_{i, j} = 
     \left\{
     \begin{array}{lr}
     0 &  \text{if $i = 0 \vee j = 0$}\\
     X_{i-1, j-1} + 1 &  \text{if $i, j > 0 \wedge s_i = t_j$}\\
     \max\{X_{i, j-1}, X_{i-1, j}\} &  \text{if $i, j > 0 \wedge s_i \neq t_j$}
     \end{array}
     \right.
    \label{eq:lcs}
\end{align}
For simplicity, we assume that the two input sequences are of the
same length, i.e. $n = m$.

\begin{figure}
\hspace*{-1.5cm}
\begin{minipage}[b]{0.43 \linewidth}
\centering
\begin{tikzpicture}[scale=.9, transform shape]
\draw (0, 0) rectangle (4, 4);
\draw [thick, ->] (0.25, 4.25) -- (1, 4.25) node[fill=white, pos=.5]{$j$};
\draw [thick, ->] (-0.25, 3.75) -- (-0.25, 3) node[fill=white, rotate=90, pos=.5]{$i$};
\draw [barrow] (-0.5, 1.5) -- (1.5, -0.5) node [pos=.5, below, sloped]{time};
\draw [thick] (2.5, -0.5) -- (4.5, 1.5) node [pos=.5, above, sloped]{Anti-Diagonal};

\draw (0, 2) -- (4, 2);
\draw (2, 0) -- (2, 4);

\foreach \y in {2, 1}
\foreach \x in {1,...,\y}{
    \draw [thin] (2*\x-2, 4-2*\y + 2*\x-1) -- (2*\x, 4-2*\y + 2*\x-1); 
    \draw [thin] (2*\x-1, 4-2*\y + 2*\x-2) -- (2*\x-1, 4-2*\y + 2*\x);
};
\foreach \x in {0.5, 1.5, 2.5, 3.5} {
    \node at (\x, \x) {$1$};
};

\foreach \y in {3, 2, 1}
\foreach \x in {1,...,\y}{
    \draw [dashed] (\x-1, 3-\y + \x+0.5) -- (\x, 3-\y + \x+0.5);
    \draw [dashed] (\x-0.5, 3-\y + \x) -- (\x-0.5, 3-\y + \x+1);
};

\foreach \y in {1, 2, 3}
\foreach \x in {1, 2, 3, 4}{
    \node at (0.5*\x - 0.25, 0.5*\y + 0.5*\x + 0.25) {$2$};
};

\foreach \s in {0.25}
\foreach \y in {3, 2, 1}
\foreach \x in {1,...,\y}{
    \draw [dashed] (2*\s*\x-2*\s, 2.5-\s+2*\s*\x + 6*\s-2*\s*\y) -- (2*\s*\x, 2.5-\s+2*\s*\x + 6*\s-2*\s*\y);
    \draw [dashed] (2*\s*\x-\s, 2.5+2*\s*\x-2*\s + 6*\s-2*\s*\y) -- (2*\s*\x-\s, 2.5+2*\s*\x + 6*\s-2*\s*\y);
};

\foreach \s in {0.125}
\foreach \y in {1, 2, 3}
\foreach \x in {1, 2, 3, 4}{
    \node at (2*\s*\x-\s, 2.5+2*\s*\x+2*\s*\y-3*\s) {$3$};
};

\draw [dashed] (0.5, 1.25) -- (1, 1.25);
\draw [dashed] (0.75, 1) -- (0.75, 1.5);

\draw [dashed] (1.5, 2.25) -- (2, 2.25);
\draw [dashed] (1.75, 2) -- (1.75, 2.5);

\draw [dashed] (2, 3.25) -- (3, 3.25);
\draw [dashed] (2, 3.75) -- (3, 3.75);
\draw [dashed] (2.25, 3) -- (2.25, 4);
\draw [dashed] (2.75, 3) -- (2.75, 4);

\foreach \s in {0.125}{
    \path (0.5+\s, 1+\s) node {$3$} -- (0.5+\s, 1.5-\s) node {$3$} -- (1-\s, 1.5-\s) node {$3$};
    \path (1.5+\s, 2+\s) node {$3$} -- (1.5+\s, 2.5-\s) node {$3$} -- (2-\s, 2.5-\s) node {$3$};
    \path (2+\s, 3+\s) node {$3$} -- (2.25+\s, 3+\s) node {$3$} 
        -- (2.25+\s, 3.25+\s) node {$3$} -- (2.5+\s, 3.25+\s) node {$3$}
        -- (2.5+\s, 3.5+\s) node {$3$} -- (2.75+\s, 4-\s) node {$3$};
};
\end{tikzpicture}
\caption{\algName{} LCS ($p=4$)}
\label{fig:paco-lcs-pic}
\vfil
\begin{tikzpicture}[scale = 2]
\draw [thick, ->] (0.75, 2.2) -- (1.25, 2.2) node[fill=white, pos=.5]{$x$};
\draw [thick, ->] (0.1, 1.25) -- (0.1, 0.75) node[fill=white, rotate=90, pos=.5]{$y$};
\path[draw] (0, 2) -- (2, 2) -- (2, 0) -- cycle;
\draw (1, 1) rectangle (2, 2);
\filldraw [fill=gray!60] (0, 2 - 0.01) rectangle (2, 2 + 0.01);
\node at (0.5, 2.1) {$(0, 0)$};
\node at (1.5, 2.1) {$(0, 1)$};
\node at (1.3, 0.45) {$(1, 1)$};
\draw [dashed] (1.33, 1) -- (1.33, 2);
\draw [dashed] (1.33, 1.5) -- (2, 1.5);
\draw (0.5, 1.5) rectangle (1, 2);
\draw (1.5, 0.5) rectangle (2, 1);
\draw [dashed] (0.667, 1.5) -- (0.667, 2);
\draw [dashed] (0.667, 1.75) -- (1, 1.75);
\foreach \h in {1.8}
\foreach \s in {.05}
\foreach \r in {0.01}
\foreach \x in {0.3} {
\fill[black] (\x, \h) circle (\r)
    -- (\x+\s, \h) circle (\r)
    -- (\x+2*\s, \h) circle (\r);
};
\foreach \h in {1.3}
\foreach \s in {.05}
\foreach \r in {0.01}
\foreach \x in {0.8} {
\fill[black] (\x, \h) circle (\r)
    -- (\x+\s, \h) circle (\r)
    -- (\x+2*\s, \h) circle (\r);
};

\draw [dashed] (1.667, 0.5) -- (1.667, 1);
\draw [dashed] (1.667, 0.75) -- (2, 0.75);
\foreach \h in {0.8}
\foreach \s in {.05}
\foreach \r in {0.01}
\foreach \x in {1.3} {
\fill[black] (\x, \h) circle (\r)
    -- (\x+\s, \h) circle (\r)
    -- (\x+2*\s, \h) circle (\r);
};
\foreach \h in {0.3}
\foreach \s in {.05}
\foreach \r in {0.01}
\foreach \x in {1.8} {
\fill[black] (\x, \h) circle (\r)
    -- (\x+\s, \h) circle (\r)
    -- (\x+2*\s, \h) circle (\r);
};

\end{tikzpicture}
\caption{\algName{} 1D ($p=3$)}
\label{fig:paco-1D-pic}
\vfil
%
\begin{tikzpicture}[scale = 1]

\node[above, yshift=-0.8mm] at (0, 0, 0) {$(n, n)$};
\node[above, yshift=-0.8mm] at (2, 0, 0) {$(n, 0)$};
\node[below] at (0, 0, 2) {$(0, n)$};
\node[below] at (2, 0, 2) {$(0, 0)$};

\begin{scope}[xshift=50, yshift=90]
\draw[->] (0, 0, 0) -- (-0.5, 0, 0) node [anchor=east] {$x$};
\draw[->] (0, 0, 0) -- (0, 0, -0.5) node [anchor=south west]{$y$};
\draw[->] (0, 0, 0) -- (0, 0.5, 0) node [anchor=south]{$z$};
\node[below] at (0, 0, 0) {$o$};
\end{scope}

\begin{scope}[]
\path (0, 0, 0) pic {bbox} (0, 1, 0) pic {bbox} (0, 2, 0) pic {tgap};
\path (1, 0, 0) pic {rbox} (1, 1, 0) pic {rgap};
\path (0, 0, 1) pic {lbox} (0, 1, 1) pic {lgap};
\path (1, 0, 1) pic {bgap};
\end{scope}

\begin{scope}[xshift=55]
\path (1, 0, 1) pic {fbox} (1, 1.5, 1) pic {bgap};
\end{scope}

\node [right] at (0, 0, 0.5) {$11$};
\node at (1.5, 0, 0.5) {$10$};
\node [right] at (0, 0, 1.5) {$01$};
\node at (1.5, 0, 1.5) {$00$};

\end{tikzpicture}
\caption{The work of GAP}
\label{fig:paco-gap-pic}
\end{minipage}
\hfill
%
\begin{minipage}[b]{0.55 \linewidth}
\centering
\input{paco-1D-code}
\caption{\algName{} 1D code. $\proc{co-1D}_{\Box}$ is sequential,
and $\proc{cop-1D}_{\Box}$ is parallel.}
\label{fig:paco-1D-code}
\end{minipage}
\end{figure}

\begin{lemma} [\cite{ChowdhuryRa06}]
    There is a sequential algorithm \proc{co-lcs} that 
    computes the LCS recurrences of \eqref{lcs} in optimal 
    $O(n^2)$ work, using no temporary space, and 
    $O(n^2/(LZ) + n/L)$ cache misses in a cache-oblivious 
    fashion.
\label{lem:co-lcs}
\end{lemma}


Referring to \figref{paco-lcs-pic}, we design a two-phase
\ourAlg{LCS}{} as follows.
Firstly, a partitioning phase divide-and-assigns
regions to $p$ processors evenly as follows. 
Initially, the entire 2D square region is marked as ``unassigned''.
It then repeatedly makes a $2$-way division on all unassigned 
sub-regions. As soon as it finds some anti-diagonal, i.e.
all sub-regions on the same anti-diagonal have their center 
coordinates $(i, j)$ satisfying that $i+j$ are equal,
contains equal or more than $p$ sub-regions, it assigns $p$ of 
them to $p$ processors in a round-robin fashion and stops any 
further division on them. If the sub-regions on an anti-diagonal
are of constant size, it assigns all of them to $p$ 
processors in a round-robin fashion.
\afigref{paco-lcs-pic} is an illustrative diagram
, assuming $p = 4$; Labels of sub-regions denote the order
they get assigned. For example, label-$1$ sub-regions are
the firstly assigned sub-regions after two rounds of $2$-way
division; label-$2$ sub-regions require one more round of 
$2$-way division, and so on.
Secondly, the algorithm executes 
sub-regions anti-diagonal by anti-diagonal along a time line.
All sub-regions on the same anti-diagonal run simultaneously. 
Each sub-region is executed sequentially by the
best sequential cache-oblivious algorithm \cite{ChowdhuryRa06}
(\lemref{co-lcs}).
Since each sub-region only depends on two of its 
neighboring regions, there is no need of global synchronization 
between consecutive anti-diagonals. In semantics, the data 
dependency between sub-regions can be specified by
a dataflow operator like the $\fire$ operator in the 
Nested Dataflow Model \cite{DinhSiTa16}.

\begin{theorem}
    The \ourAlg{LCS}{} computes the LCS recurrences of 
    \eqref{lcs} in optimal $\tsum = O(n^2)$ work,
    $\tinf = O(n^2/p)$ time, using no temporary space,
    $\qsum = O(\min\{n^2/(LZ) + (p n \log (pZ)) / L, (p n \log n)/L\})$
    and
    $\qinf = O(\min\{n^2/(pLZ) + (n \log (pZ))/L, (n \log n)/L\})$
    , assuming $p = o(n)$ if does not count partitioning 
    overheads. 
\label{thm:lcs}
\end{theorem}
\begin{proof}
%
\punt{
\paragrf{Correctness: }
All sub-regions on the same anti-diagonal do not
have data dependency on each other, hence can run simultaneously.
All partitioning, assigning and execution are based on 
anti-diagonals along a time line.
}

%
Our performance analyses consider only 
the execution phase, with partitioning overheads calculated
separately in \corref{lcs-part}.
The work and space complexities come from the fact that this 
algorithm calls the sequential algorithm 
(\lemref{co-lcs}) to compute each assigned sub-region.

\paragrf{\Obc{}: } 
Clear from the partitioning phase because in each assignment
the difference in work loads among processors can not
be more than an asymptotically smaller term (normal assignment) 
or a small constant (base-case assignment).
\punt{
In the partitioning phase, a sub-region gets 
assigned either because the number of sub-regions on the same 
anti-diagonal is larger than or equal $p$ or because all 
sub-regions on the same anti-diagonal are of constant size.
In the first case, the algorithm assigns exactly $p$ of 
them so that the difference in workloads among different 
processors will not be more than one row or column, i.e. an 
asymptotically smaller term. 
In the second case, the number of base cases can not be more 
than $2p$ because of the $2$-way division so that
the difference can not be more than a constant.
Now we can conclude that in either case the time complexity,
also known as span, depth, or critical-path length, is $O(n^2/p)$.
}

\paragrf{\Occ{}: }\\
\paragrf{Outline: }
We firstly prove that the sub-regions assigned to each processor 
form an almost geometrically decreasing sequence in terms of area.
Since we apply the sequential cache-oblivious algorithm
(\lemref{co-lcs}) to compute each sub-region
and the sequential cache complexity is proportional to area,
top-level sub-regions thus dominate. Summing up over all 
sub-regions on each and every processors then yields the bounds. 
%

\paragrf{More details: } %
We prove by induction that the sub-regions assigned to
each processor form an almost geometrically decreasing sequence
in terms of area.
Referring to \figref{paco-lcs-pic}, the labels $1$, $2$, 
and $3$ denote the order a sub-region gets assigned. 
To simplify analysis, we assume without loss of generality 
that the the first step of partitioning makes a 
$p$-way division on the entire region.
\punt{
After the first step, 
there are exactly $p$ label-$1$ largest sub-regions of dimension
$(n/p) \pm 1$-by-$(n/p) \pm 1$ on the middle anti-diagonal,
which partitions the entire region to a top-left and a 
bottom-right parts. 
The $p$ label-$1$ sub-regions will be assigned to $p$ processors
in a round-robin fashion.
Both top-left and bottom-right parts will then have $(p-1)$ 
anti-diagonals, each of which has $1$ to $(p-1)$ sub-regions of 
dimension $(n/p) \pm 1$-by-$(n/p) \pm 1$, respectively.
}
\punt{
Let's focus on the top-left part since the 
bottom-right is just a mirroring. 
Performing one more level of $2$-way division on the top-left
$(p-1)$ anti-diagonals will yield $2 (p-1)$ new anti-diagonals,
$(p-1)$ of which will have $p$ label-$2$ sub-regions of 
dimension $(n/(2p)) \pm 1$-by-$(n/(2p)) \pm 1$.  
So, each processor will be assigned $(p-1)$ label-$2$ sub-regions.
These label-$2$ sub-regions again partition the top-left part 
into a (top-left)-(top-left) and a (top-left)-(bottom-right) 
part, each of which again has $(p-1)$ anti-diagonals and
each anti-diagonal has $1$ to $(p-1)$ 
sub-regions of dimension $(n/(2p)) \pm 1$-by-$(n/(2p)) \pm 1$
, respectively. This recursive divide-and-assign procedure repeats 
until all sub-regions are of constant size, in which case they are
assigned in a round-robin fashion to $p$ processors.
}
Then by a recursive $2$-way divide-and-assign, except the 
top-level label-$1$ sub-regions, if each processor has $q$ 
label-$j$ sub-regions, there will be $2q$ label-$(j+1)$ sub-regions
on the same processor. The total area of all label-$j$ and
label-$(j+1)$ sub-regions on each processor sum up to 
$q (n/2^{j})^2$ and $2q (n/2^{j+1})^2$, respectively.
In conclusion, the sum of label-$j$ areas
is a factor of $2$ larger than that of label-$(j+1)$'s,
which then forms a geometrically decreasing sequence for all 
$j \in [1, \log n]$. 
On the other hand, we can see that the sums of half-perimeter,
which stands for the space requirement of sub-regions,
of consecutively labelled sub-regions are identical.

Applying the sequential algorithm of \lemref{co-lcs} to compute
each assigned sub-region, assuming that $n/p > \epsilon Z$,
where $\epsilon \in (0, 1)$ is some small constant, the maximal 
cache misses on any single processor sums up to 
$Q_{1, \proc{co-lcs}}(n/p) + \sum_{i=2}^{j} 2^{i-1} (p-1) 
\cdot Q_{1, \proc{co-lcs}}(n/(2^{i-1}p)) + 
\sum_{i=j+1}^{\log n} 2^{i-1} (p-1) \cdot Q_{1, \proc{co-lcs}}(n/(2^{i-1}p)) 
= O(n^2/(pLZ) + (n \log (pZ))/L)$
Note that $Q_{1, \proc{co-lcs}}(n) = O(n^2/(LZ))$ 
if $2n > \epsilon Z$ and $O(n/L)$ if $2n \leq \epsilon Z$ by
\lemref{co-lcs} \cite{ChowdhuryRa06}
\footnote{The input array of LCS are stored in align with
anti-diagonal, so its total input size is $2n$.}, which 
explains the first equation.
The switching point $j$ comes when $2n/(2^j p) \leq \epsilon Z$
, i.e. when the input array size is less than or equal 
$\epsilon Z$, which solves to $j \geq \log_2 (2n/(pZ))$. 
We make $j = \log_2 (2n/(pZ))$ to get the final bound.
To complete the calculation, if we consider the case
of $n/p \leq \epsilon Z$, 
\aeqref{paco-lcs-asym-ineq1} will reduce to 
$\qinf = O(n/(pL) + (n \log n)/L) = O((n\log n)/L)$, 
for $p \geq 1$ and $n \geq 2$,
because the sums of half-perimeters of sub-regions of 
consecutive labels are identical.
Note that in this case we have
$\log n = O(\log (pZ))$ so that we must take a $\min$, rather
than a $\max$, over the two cases to yield an overall bound.
Since this is the analysis for any single processor, it is 
then clear $\qsum = p \qinf$.
This finalizes the proof for \occ{}.
\end{proof}

\begin{corollary}
The partitioning overheads of \ourAlg{LCS}{} are $O(p^2 n)$.
The overheads are asymptotically smaller than the computational 
loads assigned to any single processor if $p = o(n^{1/3})$.
\label{cor:lcs-part}
\end{corollary}
\begin{proof}
    The partitioning overheads are proportional to the number
    of leaves of the pruned binary tree of algorithm.
    According to the proof of \occ{} of \thmref{lcs}, except
    the top-level label-$1$ regions, if each 
    processor has $q$ label-$j$ regions, there will be $2q$ 
    label-$(j+1)$ regions on the same processor.
    So we can bound the number of total leaves by 
    $p [1 + \sum_{i=2}^{\log n} 2^{i-1} (p-1)] = O(p^2 n)$.
    Compared with the computational loads
    assigned to any processor, which is $O(n^2/p)$, it is
    asymptotically smaller if $p = o(n^{1/3})$.
\end{proof}

A Nested Parallel, which has a series-parallel DAG, algorithm
scheduled by a Randomized Work-Stealing (RWS) scheduler such 
as Cilk
will yield $O(p T_\infty)$ steals \cite{AcarBlBl00} with high
probability, which are its partitioning overheads.
The PO LCS algorithm \cite{FrigoSt09} will then have a 
partitioning overheads of
$O(p T_\infty) = O(p n^{\log_2 3})$, which is asymptotically
larger than ours if $p = o(n^{\log_2 3 - 1}) \approx o(n^{0.58})$.
Compared with the PA LCS \cite{ChowdhuryRa08} that has an 
$O(p^2)$ overheads, our overheads are larger due to more 
sub-regions generated. We leave an efficient parallelization of 
\ourAlg{LCS}{}'s partitioning phase to future research.

\begin{corollary}
    The \ourAlg{LCS}{} can achieve \perfectspeedup{} if
    $n/p = \Omega(Z \log (pZ))$, if does not count partitioning 
    overheads.
    \label{cor:lcs-scale}
\end{corollary}
\begin{proof}
    The \ourAlg{LCS}{} has memory-dependent bound of 
    $\qsum = O(n^2/(LZ) + (p n \log (pZ))/L)$ 
    if $n/p > \epsilon Z$ and memory-independent bound of
    $\qsum = O((p n \log n)/L)$ if $n/p \leq \epsilon Z$.
    It is then clear that \perfectspeedup{} comes when the 
    memory dependent bound holds and its second
    term be subsumed by the first term.
\end{proof}

\paragrf{Discussions: }
The classic PO and cache-efficient LCS algorithm
\cite{FrigoSt09, CormenLeRi09} has a critical-path length of 
$O(n^{\log_2 3})$, which induces a parallel cache complexity of 
$Q_p = O(n^2/(LZ) + (p n^{\log_2 3} Z)/L)$ with high probability
when scheduled by an 
RWS scheduler \cite{AcarBlBl00, BlumofeLe99, BlumofeFrJo96a}.
This bound is asymptotically larger than ours. 
Moreover, our bound is deterministic.
Later Frigo and Strumpen \cite{FrigoSt09} improved the 
bound to $Q_p = O(n^2/(LZ) + \sqrt{p n^{3.58}})$ by using a
concave function and Jensen's Inequality. 
We can see 
that if $p^{0.5} \log (pZ) = o(n^{0.78})$, which is usually true 
on any given machine whose $p$ and $Z$ are constants 
with respect to 
problem size $n$, our bound can still be asymptotically smaller.
Cole and Ramachandran \cite{ColeRa12a} later pointed out that 
Frigo and Strumpen's method may omit the overheads
of usurpation, i.e. synchronization at the join point of a 
fork-join (also known as nested parallel) algorithm.
They gave a refined overall
cache bound of $O(n^2/(LZ) + \sqrt{p n^{3.58}} + p n^{1.58})$ 
for finding LCS sequence, more than just the length, if 
approximating $\log_2 3 \approx 1.58$.
On the other hand, Chowdhury and Ramachandran \cite{ChowdhuryRa08}
designed cache-efficient LCS algorithms for several different
models, including D-CMP, S-CMP, and Multicore. Their D-CMP model
is exactly the \ourModel{} \cite{FrigoSt09} adopted by our
paper. Their LCS algorithm on the D-CMP model makes a $p$-way 
divide-and-assign at the
top level of recursion then switches to the sequential 
$2$-way divide and conquer (\lemref{co-lcs}) for the rest
of computation. The bound claimed in their paper considers 
only the case when $n/p > \epsilon Z$. If we consider both 
branches, their bound will then be
$\qinf = O(n^2/(pLZ))$ and $\qsum = O(n^2/(LZ))$ if 
$n/p > \epsilon Z$; and $\qinf = O(n/L)$ and $\qsum = O(pn/L)$
if $n/p \leq \epsilon Z$.
If $\epsilon Z < n/p < \epsilon Z \log (pZ)$ or 
if $n/p \leq \epsilon Z$,
their bound will be a logarithmic factor smaller 
than ours in either case; otherwise, the two bounds are identical.
The difference is because their algorithm derives less number of
independent sub-regions.
Their algorithm's critical-path length is 
$(2p-1)n^2/p^2 + o(n^2/p) = 2 n^2/p + o(n^2/p)$, 
which is larger than our $n^2/p + o(n^2/p)$ by a small constant
factor of $2$. In practice, constant factor matters. 
Our preliminary experimental results (see our online full 
version) show that their algorithm's real performance 
is not as good as ours.

\subsecput{1D}{\ourAlg{1D}}

Given a real-valued function $\proc{w}(\cdot, \cdot)$, which 
can be computed with no memory access in $O(1)$
time, and initial value $D[0]$, compute
\begin{align}
D[j] &= \min_{0 \leq i < j} \{D[i] + w(i, j)\} & & \text{for $1 \leq j \leq n$} 
\label{eq:1D}
\end{align}
This problem was called the least weight subsequence (LWS)
problem by Hirschberg and Larmore \cite{HirschbergLa87}.
We will call it \emph{1D} problem following the convention
of Galil and Park \cite{GalilPa94} since it is a 1D simplication
of the more complicated GAP problem (\secref{gap}).
Its applications include, but is not limited to, the optimum
paragraph formation and finding a minimum height B-tree.

\begin{lemma} [\cite{ChowdhuryRa06}]
    There is a sequential external-updating function 
    $\proc{co-1D}_{\Box}$ that
    computes a rectangular quadrant of 1D problem in optimal
    $O(n^2)$ work, using no temporary space, and 
    $O(n^2/(LZ) + n/L)$ cache misses in a cache-oblivious fashion.
\label{lem:co-1D-box}
\end{lemma}

Referring to \figref{paco-1D-pic}, we can see that the 
geometric shape of total work of computing 1D problem is an 
equilateral right triangle (triangle in short). 
The output of algorithm overlaps the input and is marked by 
the top shaded row.
The sequential algorithm \cite{ChowdhuryRa06}, as well as a 
straightforward cache-oblivious parallelization (COP),
recursively divides the work into three
or four quadrants depending on shape and schedules their
execution according to the data dependencies in 
granularity of quadrants.
For convenience, we mark the top-left quadrant of each 
recursion by $(0, 0)$, top-right $(0, 1)$, bottom-left 
$(1, 0)$, and bottom-right $(1, 1)$.
A triangular quadrant is a 1D computation by only cells within
the same quadrant, i.e. a self-updating function, while a 
squared quadrant denotes an update of region by cells
from a disjoint quadrant, i.e. an external-updating function.
The cache-oblivious (both sequential and parallel) algorithm 
\cite{ChowdhuryRa06}
firstly invokes itself recursively on the $(0, 0)$ quadrant,
then updates the output of $(0, 1)$ by the
results of $(0, 0)$, finally recursively computes 
the $(1, 1)$, whose output overlaps that of $(0, 1)$.

Our \ourAlg{1D}{} only changes the partitioning and 
parallelization of the squared $(0, 1)$ quadrant of each 
recursion as follows.
Initially the top-level square is associatd with a list of
all $p$ processors. It then divides the square along a longer
dimension into two halves by the ratio of 
$\lfloor p / 2 \rfloor : \lceil p / 2 \rceil$. In the mean
time, it splits the processor list by the same ratio and hands
down the resulting two lists to the two halves respectively. 
If a rectangle has two equal-sized dimensions, division
can be on an arbitrary one to break tie.
If a division is on the $y$ axis (\figref{paco-1D-pic}),
the algorithm will 
allocate temporary space to break dependency since the
two resulting rectangles update the same output region.
In this case, the two resulting rectangles will merge the results
concurrently after both of them have finished local computation.
The divide-and-conquer procedure of each squared $(0, 1)$ 
quadrant of each recursion repeats until each derived rectangle 
is associated with a list of only one ($1$) processor, specifying 
on which the computation of rectangle will be executed 
sequentially.
The partitioning and parallelelization of squared quadrant
will apply recursively to the triangular $(0, 0)$ and $(1, 1)$ 
quadrants of every recursions until base cases. A base case 
will be executed sequentially on an arbitrary processor.
\figref{paco-1D-pic} shows a diagram assuming $p = 3$ and
\figref{paco-1D-code} is the pseudo-code.
In \figref{paco-1D-code}, $\proc{cop-1D}_{\tri}$ denotes the
self-updating function, $\proc{cop-1D}_{\Box}$ the 
parallel external-updating function, and $\proc{co-1D}_{\Box}$
the sequential external-updating function.


\begin{theorem}
    The \ourAlg{1D}{} computes the 1D recurrence of
    \eqref{1D} in optimal $\tsum = O(n^2)$ work,
    $\tinf = O(n^2/p)$ time,
    using $O(p^{1/2} n \log n)$ temporary space,
    $\qsum = O(\min \{n^2/(LZ) + (p Z \log Z) / L, (p^{1/2} n \log n) / L\})$
    and
    $\qinf = O(\min \{n^2/(pLZ) + (Z \log Z) / L, (n \log n) / (p^{1/2} L)\})$, assuming $p = o(n)$. The \perfectspeedup{} range 
    is $n = \Omega(Z \sqrt{p \log Z})$.
\label{thm:1D}
\end{theorem}
\begin{proof}
The work and time complexity bounds follow from that
the algorithm always evenly partitions the square of each and
every recursions among $p$ processors until base cases.
Chowdhury and Ramachandran \cite{ChowdhuryRa06} 
(\lemref{co-1D-box}) showed that 
the sequential external-updating function incurs 
$O(n^2/(LZ) + n/L)$ cache misses on a square of dimensions
$n$-by-$n$, which indicates that the cache complexity 
is proportional to the area, i.e. $O(n^2)$, if its space 
requirement $2n$ is larger than cache size $Z$, otherwise 
proportional to the half-perimeter, i.e. $O(n)$. 
Note that the space requirement of an
external-updating function is the half-perimeter of square.
The width along $x$ axis (\figref{paco-1D-pic}) stands
for the output region and the length along $y$ axis 
for the input.
So we just need to count the
areas and half-perimeters of the rectangles assigned to
each processor to bound the $\qsum$ and $\qinf$.
Since the partitioning always divides a rectangle with $p'$ 
processors into two halves by the ratio of 
$\lfloor p' / 2 \rfloor : \lceil p' / 2 \rceil$, the area ratio
of any final rectangle derived from an initial squared $(0, 1)$ 
quadrant is clearly $1/p$. Applying the conclusion 
recursively to all triangles of every recursions yields an 
$O(n^2 / p)$ total area on each processor.
We take two steps to bound the half-perimeter of each rectangle
as follows.
Firstly, we prove the 
bound by assuming that $p$ is a power of two. Secondly, we 
prove that the resulting cache complexity will not differ by a 
small constant factor when removing the assumption.
The initial half-perimeter of an $n$-by-$n$ square is $2n$,
and we use notation $S_p^{+}$ to denote the overall increase of 
half-perimeters after $\lceil \log_2 p \rceil$ rounds of $2$-way
division.
\begin{enumerate}
    \item If $p$ is a power of two: In this case, the algorithm
        cuts the initial square alternatively on the two 
        dimensions into two equal-sized halves.
        So the division doubles the initial 
        half-perimeter of $2n$ every two rounds. That is,
        $S_p^{+} = \sum_{i = 0}^{(1/2) \log_2 p} (2n \cdot 2^i) 
        \leq 4 p^{1/2} n = O(p^{1/2} n)$.

        The overall half-perimeter is then $2n + S_p^{+} = 
        O(p^{1/2} n)$ and the half-perimeter of each rectangle
        will be $O(n/p^{1/2})$ because all final rectangles 
        are of the same shape and size.

    \item If $p$ is not a power of two: This time the 
        algorithm may cut a rectangle into
        two slightly unequal-sized halves. 
        For simplicity of analysis, we assume 
        that it follows the same partitioning
        order on every dimensions as in the case of rounding $p$
        up to the next power of two. We can then bound any
        dimension of any final rectangle to be no more than a
        small constant factor away from that in the case of
        rounding $p$ to the next power of two.
        We take an arbitrary dimension of length $n$ as an 
        example. In the worst case, the dimension gets cut through
        a series of uneven right-halves (uneven left-halves 
        are similar and symmetric) and will have size
        $n \cdot \frac{\lceil p/2 \rceil}{p} \cdot 
        \frac{\lceil (\lceil (\lceil p/2 \rceil) / 2 \rceil) / 2 \rceil}{\lceil (\lceil p/2 \rceil) / 2 \rceil} \cdots
        = n \cdot \prod_{j = 0}^{\frac{1}{2} \lceil \log_2 p \rceil} 
        \frac{2^{2j + 1} + 1}{2^{2j + 2} + 1} 
        = \Theta(n/p^{1/2})$, which is asymptotically the same 
        as cutting through a series of even divisions.
        The number $\frac{1}{2} \lceil \log_2 p \rceil$ is because 
        the algorithm cuts alternatively on the two dimensions 
        and total rounds of cutting is $\lceil \log_2 p \rceil$.
        The equation holds because
        $\prod_{i = 0}^{x} \frac{2^i + 1}{2^{i+1} + 1} 
        = \Theta(2^{-x})$ and $\forall j \in [0, x], 
        \frac{2^{2j+1} + 1}{2^{2j+2} + 1} 
        \leq \frac{2^{2j} + 1}{2^{2j+1} + 1}
        \leq \frac{2^{2j-1} + 1}{2^{2j} + 1}$,
        so $\prod_{j = 0}^{x/2} \frac{2^{2j+1} + 1}{2^{2j+2} + 1}
        = \prod_{j = 0}^{x/2} \frac{2^{2j} + 1}{2^{2j+1} + 1}
        = \Theta(2^{-x/2})$, where $x = \lceil \log_2 p \rceil$.
\end{enumerate}
Combining the above two cases, we conclude that the area and 
half-perimeter of each final rectangle of the top recursion
is $O(n^2/p)$ and
$O(n/p^{1/2})$ respectively. Applying \lemref{co-1D-box}
will yield a cache complexity of $O(n^2/(pLZ) + n/(p^{1/2}L))$
for each top-level rectangle assigned to each processor.
Note that when the algorithm cuts a rectangle on the $y$ axis
into two halves, it will merge the results after the two 
halves have finished their local computation. Since the merge
is just one row of a rectangle and can be fully parallelized
among the processor list of the parent rectangle as shown
by \lirefs{parForMergeBegin}{parForMergeEnd} in 
\figref{paco-1D-code}, we can charge its overheads
to the two halves without affecting asymptotically on
either computation or communication bounds.
From the algorithm, we can see that going down one more level
of recursion will double the number of rectangles assigned to 
each processor, shrinks the corresponding total area by a factor
of $2$, and keeps the same total half-perimeter.
So if $n > \epsilon Z$, where $\epsilon$ is some small constant,
$\qinf = O((n/2)^2 / (pLZ) + 2 \cdot ((n/4)^2 / (pLZ))+ \cdots 
+ Z/L + Z/L + \cdots) = O(n^2/(pLZ) + (Z \log Z) / L)$.
If $n \leq \epsilon Z$, $\qinf = O(n/(2 p^{1/2} L) 
+ 2 \cdot (n/(4 p^{1/2} L)) + \cdots) 
= O((n \log n) / (p^{1/2} L))$.
$\qsum = p \qinf = O(\min\{n^2/(LZ) + (p Z \log Z) / L, 
(p^{1/2} n \log n) / L\})$.
The overall temporary space is the sum of half-perimeters
over all derived rectangles, which is $O(p^{1/2} n \log n)$.

The perfect strong scaling range comes when $n > \epsilon Z$ and
the second term of $\qsum$, i.e. $O(p Z \log Z / L)$ is subsumed
by the first term, i.e. $O(n^2/(LZ))$.
\end{proof}

Note that the partitioning overheads of \ourAlg{1D}{} is 
proportional to the number of rectangles assigned to each
and every processors, so is charged to computational loads.

\paragrf{Discussion: }
The PO 1D algorithm developed by Chowdhury and Ramachandran 
\cite{ChowdhuryRa06} has a sequential cache complexity of 
$O(n^2/(LZ) + n/L)$ with a depth of $O(n \log n)$. So a 
straightforward scheduling by a Randomized Work-Stealing (RWS) 
scheduler will yield
a parallel cache complexity of $O(n^2/(LZ) + (p n \log n Z) / L)$,
which is asymptotically larger than our bound.
Blelloch and Gu \cite{BlellochGu18} improved the depth to 
$O(n)$ by allocating $O(p^{1/2} n)$ total 
temporary space from an arbitrarily large system's stack. 
Their algorithm's parallel cache 
complexity, assuming an RWS scheduler, is 
$O(n^2/(LZ) + (p n Z)/L)$, which is still asymptotically larger
than ours in both the case $n > \epsilon Z$ and 
$n \leq \epsilon Z$.
Galil and Park \cite{GalilPa94} developed a sublinear 
$O(\sqrt{n} \log n)$-depth 1D algorithm, which requires a 
sub-optimal $O(p^{1/3} n^{3/2})$ total space and
$O(n^2/L)$ sequential cache complexity. This bound is 
the largest of all above algorithms.

\subsecput{gap}{\ourAlg{GAP}}
Given $w$, $w'$, $s_{ij}$, which can be computed in
$O(1)$ time with no memory access, and $D_{0, 0} = 0$,
compute
\begin{align}
    D_{i, j} = \min
    \left\{
    \begin{array}{l}
    D_{i - 1, j - 1} + s_{ij} \\
    \min_{0 \leq q < j} \{D_{i, q} + w(q, j)\} \\
    \min_{0 \leq p < i} \{D_{p, j} + w'(p, i)\} 
    \end{array}
    \right.
\label{eq:gap}
\end{align}
for $0 \leq i \leq m$ and $0 \leq j \leq n$.
We assume that $m$ and $n$ are equal to simplify discussion.
This is the problem of computing edit
distance when allowing gaps of insertions and deletions
\cite{GalilGi89}. We will call it \emph{GAP} problem
following the convention of Galil and Park \cite{GalilPa94}.
Its applications include, but is not limited to, molecular
biology, geology, and speech recognition.

GAP problem is actually a 2D version of 1D problem
(\secref{1D}). Similarly, the cache-oblivious algorithms, 
both sequential and a straightforward parallel version,
designed
by Chowdhury and Ramachandran \cite{Chowdhury07, ChowdhuryRa06}
follow a similar recursive divide-and-conquer pattern to 
their 1D algorithm and separate
the updates to any quadrant to one self-updating function
and one external-updating function.
The geometric shape of the work of a self-updating function is 
a 3D triangular analogue,
while that of an external-updating function is a 3D cube.
The right part of \figref{paco-gap-pic} shows such a 3D
triangular analogue on the top and a 3D cube at bottom.
\punt{
Blelloch and Gu \cite{BlellochGu18} improved the cache complexity
by observing that an external-updating function of dimension $n$
can be decomposed into $n$ independent invocations of 1D 
algorithm's external-updating function, i.e.
a 3D cube can be decomposed into a set of independent 2D 
squares. 
}

\punt{
\begin{lemma} [\cite{BlellochGu18}]
    There is a sequential external-updating function
    $\proc{co-gap}_{\Box}$ that
    computes a 3D cube of the GAP problem in 
    optimal $O(n^3)$ work, using no temporary space,
    and $O(n^3/(LZ) + n^2/L)$ cache misses in a 
    cache-oblivious fashion.
\label{lem:co-gap-box}
\end{lemma}
}

Similar to the case in 1D, 
our \ourAlg{GAP}{} only changes the partitioning of
external-updating function as follows.
It always partitions the work of a 3D cube of dimensions
$n$-by-$n$-by-$n$ into $p$ $n$-by-$n$-by-$n/p$ 
cuboids, so that each function updates a disjoint output region 
independently and simultaneously.
The same partitioning and parallelizing pattern then applies 
recursively to every self-updating functions of every recursion,
i.e. 3D triangular analogues, until base cases. A base case is
assigned to an arbitrary processor.
%

\begin{theorem}
    The \ourAlg{GAP}{} computes the GAP recurrences of
    \eqref{gap} in optimal $O(n^3)$ work, $O(n^3/p)$ time,
    using no temporary space,
    $\qsum = O(\min\{n^3/(LZ) + (n^2 \log Z) / L, (n^2 \log n) / L\})$
    and
    $\qinf = O(\min\{n^3/(pLZ) + (n^2 \log Z) / (pL), (n^2 \log n) / (pL)\})$, assuming $p = o(n)$. The \perfectspeedup{} 
    range is $n = \Omega(Z \log Z)$.
\label{thm:gap}
\end{theorem}
\begin{proof}
Similar to that of \thmref{1D}, hence omitted.
\punt{
The work, time, and space complexities come from a similar
argument to the proof of \ourAlg{1D}{} (\thmref{1D}).
We separate the analyses on cache complexity
into two cases as follows.
\begin{enumerate}
    \item If $n > \epsilon Z$, where $\epsilon$ is some small 
        constant: In this case, for the first recursion level,
        the amount of cache misses assigned to each processor is
        $4 \cdot [(n/2p) \cdot Q_{1D, \Box}(n/2)] = 
        4 \cdot [(n/2p) \cdot n^2/(4LZ)] = n^3/(2pLZ)$, where 
        $Q_{1D, \Box}$ stands for the sequential cache complexity
        of 1D algorithm's external-updating 
        function (\lemref{co-1D-box}) in the case of $n > Z$.
        By one more level of recursion, the additional amount of
        cache misses assigned to each processor becomes
        $4 \cdot 4 \cdot [(n/4p) \cdot Q_{1D, \Box}(n/4)]
        = 4^2 \cdot [(n/4p) \cdot n^2/(4^2 LZ)] = n^3/(4pLZ)$.
        We can see that it forms a geometrically decreasing
        sequence in terms of volume, i.e. $n^3$, for the first
        $\log_2 n - \log_2 Z = \log_2 (n/Z)$ levels of recursion 
        and the root dominates.

        After $\log_2 (n/Z)$ levels of recursion, the amount of 
        cache misses assigned to each processor is
        $4^{\log_2 (n/Z)} \cdot [(n/(2^{\log_2 (n/Z)} p))
        \cdot Q_{1D, \Box}(n/2^{\log_2 (n/Z)})]
        = (n^2/Z^2) \cdot [(Z/p) \cdot (Z/L)] = n^2/(pL)$.
        This is because after $\log_2 (n/Z)$ levels of recursion,
        the space for each 1D external-updating function fits
        completely in cache with no capacity miss.
        Similarly computing the cache misses after this level,
        we can see that it is an identical sequence in terms
        of area, i.e. $n^2$, for the next $\log_2 Z$ levels.
        
        In summary, if $n > Z$, the cache complexity is
        $\qinf = O(n^3/(pLZ) + n^2/(pL) \cdot \log_2 Z)$
        and $\qsum = p \qinf$.

    \item If $n \leq \epsilon Z$: we can see that the amount of 
        cache misses assigned to each processor forms an 
        identical sequence in terms of area, i.e. $n^2/(pL)$. So 
        a total cache complexity in this case is 
        $\qinf = O(n^2/(pL) \cdot \log_2 n)$
        and $\qsum = p \qinf$.
\end{enumerate}
}
\end{proof}

\paragrf{Discussion: }
The PO GAP algorithm designed by Chowdhury and Ramachandran
\cite{ChowdhuryRa06} has a sequential cache complexity of
$O(n^3/(L\sqrt{Z}) + n^2/L)$ and a depth of $O(n^{\log_2 3})$,
using no temporary space.
So scheduling by a Randomized Work-Stealing (RWS) scheduler 
will yield a parallel cache complexity of $O(n^3/(L\sqrt{Z}) + p 
n^{\log_2 3} Z / L)$, which is asymptotically larger than ours.
Blelloch and Gu \cite{BlellochGu18} improved the sequential cache
complexity to $O(n^3/(LZ) + n^2/L \cdot \min \{\log_2 n/\sqrt{Z}, 
\log_2 \sqrt{Z}\})$ with the same $O(n^{\log_2 3})$ depth
by observing that one GAP algorithm's external-updating function 
of dimension $n$ can be decomposed into $n$ independent 
invocations of 1D algorithm's external-updating function, i.e.
a 3D cube can be decomposed into a set of independent 2D 
squares, and by allocating $O(p^{1/2} n^2)$ total
temporary space from an arbitrarily large system's stack.
Their algorithm's parallel cache
complexity, assuming an RWS scheduler, is 
then $O(n^3/(LZ) + n^2/L \cdot \min \{\log_2 n/\sqrt{Z}, \log_2 \sqrt{Z}\} + (p n^{\log_2 3} Z) / L)$, which can be slightly
smaller than ours if
$p = o((n^{0.415} \log_2 \sqrt{Z}) / Z)$, where 
$0.415 \approx 2 - \log_2 3$.
This is because our algorithm always partitions a 3D cube evenly
and recursively until base cases so incurs deviations from the
sequential execution order until base cases, while Blelloch and 
Gu's counts the sequential cache misses so there is no deviation
when the sum of input and output of a quadrant fits in cache. 
Galil and Park \cite{GalilPa94} developed a sublinear
$O(\sqrt{n} \log n)$-depth GAP algorithm, which has a
sub-optimal $O(n^4)$ work, $O(p^{1/3} n^3)$ temporary space,
and $O(n^4/L)$ sequential cache miss complexity. This bound is
the largest of all above algorithms.

\subsecput{mm}{\ourAlg{MM}}
This section considers the general rectangular MM of
multiplying an $n$-by-$k$ matrix $A$ with an $k$-by-$m$ matrix
$B$, i.e. $C = A \otimes B$, on a closed semi-ring 
$SR = (S, \oplus, \otimes , 0, 1)$, where $n, m, k$ are
arbitrary positive integers. 

\begin{figure}
\begin{center}
\input{paco-mm-code}
\end{center}
\caption{Pseudo-Code of \ourAlg{MM}}
\label{fig:paco-mm-code}
\end{figure}

\begin{figure}
\begin{center}
\input{paco-mm-1-piece-code}
\end{center}
\caption{Pseudo-code of \ourAlg{MM-1-Piece}}
\label{fig:paco-mm-1-piece-code}
\end{figure}

We can view the computation DAG of a general MM as a rectangular 
cuboid of size $n \times m \times k$, where the two side faces 
stand for the input matrices $A$ and $B$, and the bottom face 
stands for the output matrix $C$, respectively.
To perform a given multiplication, a processor
must access to the entries of $A$, $B$, and $C$, corresponding
to the projections onto the $n \times k$, $k \times m$, and 
$n \times m$ faces of the initial cuboid, respectively.

Frigo et al. \cite{FrigoLePr12} proposed a sequential 
cache-oblivious MM algorithm by making a recursive
$2$-way divide-and-conquer on the longest dimension of the 
cuboid until base cases. So the inital cuboid is computed by
a depth-first (DFS) traversal of the recursion tree. 

\begin{lemma} [\cite{FrigoLePr12}]
    There is a sequential algorithm \proc{co-mm} that 
    multiplies an $n$-by-$k$ matrix with an $k$-by-$m$ matrix 
    in optimal $O(nmk)$ work, with 
    $Q_1 = O(1 + (nm + nk + mk)/L + nmk/(L\sqrt{Z}))$ cache 
    misses in a cache-oblivious fashion.
\label{lem:seq-co-mm}
\end{lemma}

By contrast, we reduce a parallel MM algorithm to a pruned 
breadth-first (BFS)
partitioning of the initial cuboid among $p$ processors as 
follows.
The initial cuboid is marked as ``unassigned'' and has output
matrix $C$ as its bottom face. Then it 
repeatedly makes an even $2$-way division on the longest dimension
of all unassigned cuboids to derive twice the number of 
smaller cuboids depth by depth. That is, 
depth-$0$ has only one unassigned cuboid, depth-$1$
will have two, depth-$2$ will have four, and so on. 
If a division is on the height of a
cuboid, the algorithm will allocate a temporary space of
the same size as its bottom face for output of the upper cuboid.
The corresponding lower cuboid reuses their parent's
bottom face for output.
By allocating temporary space, all derived cuboids of the
same depth can run concurrently. This stands by the
observation that all multiplications are independent of each
other, serialization is only necessary 
when combining the intermediate results by addition.
As soon as some depth contains equal or more than $p$
unassigned cuboids, exact $p$ of them will be assigned to 
$p$ processors in a round-robin fashion. The rest of 
cuboids, if any, 
will go to the next round of division. This procedure repeats
until all cuboids on the same depth are of base (constant) sizes, 
in which case all of them will be assigned in a round-robin 
fashion.

\afigref{paco-mm-code} is the pseudo-code of algorithm. 
In the pseudo-code, we use notation $\{P\}$ to denote a processor 
list, $p_i$ to denote an individual processor, and \id{np} to 
denote the processor number.
The procedure has a
processor list which rounds up the $\id{np}$ real processors to 
the next power of two so that $p_i$ stands for a real processor 
if its index $i < \id{np}$ and for a virtual processor if 
$i \geq \id{np}$.
$\id{n\_rounds}$ and $\id{res\_p}$ stand for the number of 
BFS steps to the next assignment and number of leftover 
processors after the assignment, respectively.
\Lirefs{onePbegin}{onePend} executes the MM sequentially if the 
processor list reduces to just one processor.
\Lirefs{expandBegin}{expandEnd} adjust the processor list for
the next $\id{n\_rounds}$ if the only leftover processor is 
virtual.
\Lirefs{cutXbegin}{cutXend} is a straightforward parallelization
of the cutting-on-$X$ branch of \proc{seq-co-mm}. Note that
$\{\id{new\_P}_1\}$ and $\{\id{new\_P}_2\}$ returned from 
recursive procedure calls contain only real processors and will
be merged with redundant processors eliminated.
\Lirefs{pforBegin}{pforEnd} show that the parallel additions to
combine intermediate results will be executed on the returned
real processor list $\{new\_P\}$.

\afigref{pruned-bfs-pic} is an illustration of the algorithm
when $p = 3$.
After two rounds of $2$-way division, 
we have four ($4 > p = 3$) depth-$2$ unassigned cuboids, three 
of which will then be assigned to $p = 3$ processors in a 
round-robin fashion. 
The algorithm then repeats the divide-and-assign 
on the remaining one ($1$) unassigned cuboid until all unassigned
cuboids are of base (constant) sizes, in which case all of them 
will be assigned in a round-robin fashion.
The following \thmref{mm} bounds the algorithm's performance.

\begin{theorem}
    The \ourAlg{MM}{} multiplies an $n$-by-$k$ matrix
    $A$ with an $k$-by-$m$ matrix $B$ in optimal 
    $\tsum = O(nmk)$ work, optimal $\tinf = O(nmk/p)$ time, 
    using $O(\min\{pmk, \allowbreak \sqrt{pnmk^2}, \allowbreak 
p^{1/3}(nmk)^{2/3}\})$ temporary space,
    $\qsum = O(nmk/(L\sqrt{Z}) + (nm + nk + mk + \min\{pmk, 
    \allowbreak \sqrt{pnmk^2}, \allowbreak 
    p^{1/3}(nmk)^{2/3}\})/L)$ and 
    $\qinf = (1/p) \qsum$, 
    assuming $n \geq m \geq k$ and $p = o(n+m+k)$.
\label{thm:mm}
\end{theorem}
\begin{proof}
\paragrf{\Obc{}: }
%
%
A cuboid gets assigned either because the number of unassigned 
cuboids of the same depth
are equal or more than $p$, in which case exact $p$ of 
them will be assigned, or because all cuboids are of base 
(constant) size, 
in which case there will be no more than $2p$ of them and 
all of them will be assigned to $p$ processors in a 
round-robin fashion. In the first case, the difference between
assigned cuboids will be no more than one face, i.e. an
asymptotically smaller term, due to an even $2$-way division; 
while in the second case, the difference between assignments 
will be no more than a constant. 

\paragrf{\Occ{}: }\\
\paragrf{Outline: }
From the proof of \lemref{seq-co-mm} 
(Theorem 2.1 of \cite{FrigoLePr12}), the sequential
cache-oblivious MM algorithm \proc{co-mm}
incurs $O(nmk/(L\sqrt{Z}))$ cache misses, i.e. proportional
to the volume of cuboid, if its surface
area $nm + mk + nk > \epsilon Z$, and $O((nm + mk + nk)/L)$ cache
misses, i.e. proportional to the surface area, otherwise, 
where $\epsilon \in (0, 1)$ is some small constant.
We prove that $\forall i \in [1, p]$, the cuboids
assigned to processor-$i$ form
a geometrically decreasing sequence in terms of both volume
and surface area. It is then clear that the top-level, i.e. 
largest, cuboid on each processor dominates in either cases.
Since the reduction of a pair of upper and lower cuboids 
derived from a cut on height by addition is asymptotically 
cheaper than the corresponding upper and lower cuboids' 
multiplications, i.e. one face of a cuboid versus its volume,
plus that the reduction by addition can be fully parallelized,
we can charge all reduction overheads (work, time, caching) to 
all real processors that are involved in 
computing the upper and lower cuboid's multiplication 
without affecting overall complexities asymptotically.
It then boils down to bound the volume and surface area of
the largest cuboid on each processor to yield the final bounds.
%
To be convenient, we denote that the initial cuboid has volume 
$V = nmk$ and surface area $S = (nm + mk + nk)$ and assume 
without loss of generality that $n \geq m \geq k$ in the rest 
of proof.

\paragrf{More details: }
We prove that the cuboids assigned to any single processor 
form a geometrically decreasing sequence in terms of both 
volume and surface area as follows.
\punt{
Starting from $1$ root cuboid of size $n \times k \times m$,
each round of $2$-way division doubles the number of unassigned
cuboids and decrease the volume of each cuboid by a factor of $2$.
As soon as there are more than $p$ unassigned cuboids,
exact $p$ of them will get assigned. Remaining unassigned 
cuboids will go to more rounds of $2$-way division before they
can be assigned, so will be at least a factor of $2$ smaller 
than prior assigned cuboids.
}
By the $2$-way divide-and-assign, as soon as some depth contains
equal or more than $p$ cuboids, exactly $p$ of them will be 
assigned in a round-robin fashion. The number of rest 
unassigned cuboids, if any, will be less than $p$, and will go
to more rounds of $2$-way division before they can be assigned.
It is clear that no processor will have more than one cuboid of 
the same depth, i.e. the same non-constant volume.
This finalizes the proof of geometrical decrease in volume.
Since the algorithm always cut a cuboid on the longest dimension 
into two equally sized halves, we can see that the surface area
of a child cuboid is no more than $(2/3)$ of that  
of its parent but larger than $(1/2)$ fraction. 
That is, without loss of generality 
if we assume that a parent cuboid is $n' \times k' \times m'$
and has surface area $S' = (n'm' + n'k' + m'k')$,
assuming $n' \geq m' \geq k'$, we have 
$(1/2) S' \leq ((n'/2) \cdot m' + (n'/2) \cdot k' + 
m' \cdot k') \leq (2/3) S'$.
This finalizes the proof of geometrical decrease in surface area.

We then bound the volume and surface area of the largest cuboid 
on each processor as follows.
Each processor has its largest cuboid assigned after
$\lceil \log_2 p \rceil$ rounds of $2$-way division.
Since each round decreases the volume of a cuboid by
a factor of $2$, it is then clear that the volume of 
largest cuboid on each processor is 
$V/(2^{\lceil \log_2 p \rceil}) \in (V/(2p), V/p)$,
where $V$ is the volume of initial cuboid.
To bound the surface area, we adapt the proof on communication
cost of CARMA (Communication-Avoiding Recursive MAtrix 
Multiplication) algorithm by Demmel et al. 
(Sect. II C of \cite{DemmelElFo13}). The main difference is
that their proof assumes that processor number $p$ is an
exact power of $2$ and their algorithm is efficient by
the proof if all prime factors of $p$ can be bounded
by a small constant. By contrast, we adapt their proof 
to bound the surface area of the largest cuboid on each 
processor so that our algorithm and proof
work for an arbitrary number of processors, even when $p$ by
itself is a large prime number. 
We use notation $S_p^{+}$ to denote the overall
increase of surface area after $\lceil \log_2 p \rceil$ 
rounds of $2$-way division. 
\begin{enumerate}
    \item If $p \leq n/m$, the $2$-way division cuts only on 
        dimension $n$ (recall we assume $n \geq m \geq k$),
        the smallest face $m \times k$ gets doubled on every cut. 
        \begin{align}
            S_p^{+} &= \sum_{i=0}^{\log_2 p - 1} (mk \cdot 2^i) = O(pmk) \label{eq:s-sum-case1}
        \end{align}

    \item If $n/m < p \leq nm/k^2$, the division has
        two phases. The first phase of $\log_2 (n/m)$ rounds
        cut only on dimension $n$ and increase the total surface 
        area by
        $\sum_{i=0}^{\log_2 (n/m) - 1} (mk \cdot 2^i) = O(nk)$
        and increase the number of cuboids to $n/m$.
        After the first phase, the sizes of dimension $n$ 
        and $m$ of any cuboid are within a factor $2$ of each 
        other. So the second phase of 
        $\log_2 p - \log_2 (n/m) = \log_2 (pm/n)$ 
        rounds cut into all $n/m$ cuboids' dimensions $n$ and $m$ 
        alternatively and doubles the smallest face $m \times k$
        every two rounds.
        \begin{align}
            S_p^{+} &= O(nk) + \sum_{i=0}^{(1/2)\log_2 (pm/n)} ((n/m)(mk) \cdot 2^i) \label{eq:s-sum-case2}\\
                       &= O(nk) + O(\sqrt{pnmk^2}) = O(\sqrt{pnmk^2})\label{eq:s-sum-case2-subsume}
        \end{align}
        In \eqref{s-sum-case2-subsume}, the second term dominates 
        because $p > n/m$ in this case.

    \item If $nm/k^2 < p$, the division has three phases. The
        first phase cuts only on dimension $n$ for $\log_2 (n/m)$
        rounds and increase the surface area by $O(nk)$,
        as well as increasing the total number of cuboids to $n/m$.
        The second phase cuts on $n/m$ cuboids's dimension $n$ and
        $m$ alternatively for $2 \log_2 (m/k)$ rounds and 
        increase the surface area by
        $\sum_{i=0}^{\log_2 (m/k)} ((n/m)(mk) \cdot 2^i) = O(nm)$,
        as well as increasing the total number of cuboids to
        $nm/k^2$.
        After the second phase, all cuboids' three dimensions 
        are within a factor of $2$ of each other.
        So the third phase of $\log_2 p - \log_2 (nm/k^2) = 
        \log_2 (pk^2/(nm))$ rounds cut into all 
        $nm/k^2$ cuboids' three
        dimensions alternatively and double the smallest face 
        $k \times k$ every three rounds. 
        \begin{align}
            S_p^{+} &= O(nm) + \sum_{i=0}^{(1/3)\log_2 (pk^2/(nm))} ((nm/k^2)(k^2) \cdot 2^i)\\
                       &= O(nm) + O(p^{1/3} (nmk)^{2/3}) = O(p^{1/3} (nmk)^{2/3}) \label{eq:s-sum-case3-subsume}
        \end{align}
        In \eqref{s-sum-case3-subsume}, the second term dominates
        because $p > nm/k^2$ in this case.
\end{enumerate}
Combining the three cases by taking a $\min$, a single largest
cuboid's surface area is then
$O((1/p) \cdot (S + S_p^{+}) = O((1/p) \cdot (nm + nk + mk + 
\min\{pmk, \allowbreak \sqrt{pnmk^2}, \allowbreak 
p^{1/3}(nmk)^{2/3}\}))$
\punt{
$ = O((1/p) (nm + nk + mk +  p^{1/3} (nmk)^{2/3}))$ because 
$p < n/m$ and $p < nm/k^2$ in the first two terms of $\min$.
}
The temporary space complexity is then at most $S_p^{+}$.
This finalizes the bound on largest cuboid's surface area.
The bounds of theorem then follow.
\end{proof}

Though \ourAlg{MM}{} of \thmref{mm} is optimal in a shared-memory
setting, it can have up to $O(\log (n m k)$ cuboids on each 
processor so that its latency bound in a distributed-memory 
setting can be large. So we simplify the algorithm to
\ourAlg{MM-1-Piece}{} as follows.
The algorithm is almost identical to \ourAlg{MM}{} except
that each time it cuts a cuboid on its longest dimension into
two slightly unequal-sized halves as shown in 
\figref{paco-mm-1-piece-code}.
That is, if a cuboid is associated
with a list of $p$ processors, the algorithm will partition the 
cuboid on its longest dimension into two halves by the 
ratio of $\lfloor p/2 \rfloor : \lceil p/2 \rceil$. 
In the mean time, it splits the processor list by the same ratio.
The algorithm then repeats on the left and right halves 
concurrently and recursively until each cuboid is associated
with a list of only one ($1$) processor, which specifies its 
assignment.
To simplify analyses, we assume that the partitioning on each
dimension follows exactly the same order as in \ourAlg{MM}{}.
This assumption can be realized by associating the initial
real cuboid with a same-sized virtual cuboid. Each time the 
virtual cuboid employs \ourAlg{MM}{} to pick a dimension to
cut and the real cuboid will then cut on the same dimension 
but into two unequal-sized halves.
\punt{
To guarantee the same partitioning order on each dimension as 
in \ourAlg{MM}{}, we can associate the initial real cuboid with a
same-sized virtual cuboid (not shown in 
\figref{paco-mm-1-piece-code} for simplicity).
Each time it looks into the virtual
cuboid to decide which dimension is the longest before it divides
both real and virtual cuboids on the same dimension into 
two halves. The difference is that it divides the virtual cuboid
into two equal-sized halves and real cuboid into two slightly
unequal-sized halves. 
}
\corref{mm-1-piece} then bounds the algorithm's performance.

\begin{corollary}
    The \ourAlg{MM-1-Piece}{} multiplies an $n$-by-$k$ matrix
    $A$ with an $k$-by-$m$ matrix $L$, by having only one cuboid
    on each processor, in optimal $\tsum = O(nmk)$ work, 
    optimal $\tinf = O(nmk/p)$ time, 
    using $O(\min\{pmk, \allowbreak \sqrt{pnmk^2}, \allowbreak 
p^{1/3}(nmk)^{2/3}\})$ temporary space,
    with an 
    $\qsum = O(nmk/(L\sqrt{Z}) + (nm + nk + mk + \min\{pmk, 
    \allowbreak \sqrt{pnmk^2}, \allowbreak 
    p^{1/3}(nmk)^{2/3}\}) / L)$ and 
    $\qinf = (1/p) \qsum$, 
    assuming $n \geq m \geq k$ and $p = o(n+m+k)$.
\label{cor:mm-1-piece}
\end{corollary}
\begin{proof}
%
Since a real and virtual cuboid always cut on the same
dimension at each and every division points and the partitioning 
of virtual cuboid follows the same partitioning order of 
\ourAlg{MM}{} for 
the first $\lceil \log_2 p \rceil$ rounds, we just need to 
bound any dimension of any final real cuboid 
to be no more than a small constant factor away from 
that of corresponding virtual cuboid. 
The volume and surface area of any final real
cuboid will then also be within a constant factor
of those of corresponding virtual cuboid, i.e. the largest cuboid 
of \ourAlg{MM}{}.
Without loss of generality, we take dimension $n$ as an example.
In the worst case, the dimension gets cut through a series of 
right halves and will have size
$n \cdot (\lceil p/2 \rceil / p) \cdot (\lceil (\lceil p/2 \rceil) / 2 \rceil / (\lceil p/2 \rceil)) \cdots (2/3) = n \cdot \prod_{i=1}^{\lceil \log_2 p \rceil} ((2^i + 1)/(2^{i+1} + 1)) = 
\Theta(n/p)$, which is asymptotically the same as cutting through 
a series of equal-sized halves.
Similarly, we can bound the size of a dimension that gets cut 
through a series of left halves.
If the cuts on dimension $n$ interleaves with two other dimensions,
since the real cuboid follows exactly the same division order
as the virtual cuboid, the difference on any dimension after
$x$ cuts, where $0 \leq x \leq \lceil \log_2 p \rceil$, will
not be larger than a small constant factor.
This completes the proof.
\end{proof}

\begin{corollary}
    The \ourAlg{MM}{} and \ourAlg{MM-1-Piece}{} achieve 
    \perfectspeedup{} if $p = O((nmk)/Z^{3/2})$.
    \label{cor:mm-p-speedup}
\end{corollary}
\begin{proof}
    The \perfectspeedup{} range comes when the memory-independent
    bound of $\qsum = O((nm + nk + mk + \min\{pmk, \allowbreak \sqrt{pnmk^2}, \allowbreak p^{1/3} (nmk)^{2/3}\})/L)$ is subsumed
    by the memory-dependent bound of $\qsum = O(nmk/(L\sqrt{Z}))$.
\end{proof}

\paragrf{Discussions: }
%
A straightforward depth-$O(n)$ MM \cite{CormenLeRi09} has overall
parallel cache misses of $O(n^3/(L\sqrt{Z}) + pn(Z/L))$ with high 
probability when scheduled by a Randomized Work-Stealing (RWS) 
scheduler \cite{AcarBlBl00, BlumofeLe99, BlumofeFrJo96a}.
Frigo and Strumpen \cite{FrigoSt09} refined it to 
$O(n^3/(L\sqrt{Z}) + p^{1/3}n^{7/3}/L + pn)$ by using
concave function and Jensen's Inequality. 
They also pointed out that a 
static scheduling, i.e. PA scheduling, of a square MM can yield 
asymptotically less cache misses.
Blumofe et al. \cite{BlumofeFrJo96} designed a PO MM with
an $O(\log^2 n)$ critical-path length, and bounded 
\cite{BlumofeFrJo96a} its sequential cache misses to be 
asymptotically optimal on DAG-consistent
distributed shared memory maintained by the Backer coherence 
protocol.
Cole and Ramachandran \cite{ColeRa12a, ColeRa12} 
proved an $O(n^3/(L\sqrt{Z}) + (p \log p)^{1/3} \cdot n^2/L + 
p \log p)$ overall parallel cache complexity for a 
resource-oblivious algorithm scheduled by a centralized scheduler.
Chowdhury et al. \cite{ChowdhuryRaSi13} proposed a 
Multicore-Oblivious (MO) algorithm on a hierarchical multi-level
caching multicore (HM) model and a network-oblivious (NO)
algorithm on the D-BSP model with similar bounds.
Assuming $n = m = k$, our bounds are asymptotically tighter
than all above PO bounds because all PO bounds include a
non-constant critical-path length in their second term,
which is eliminated by our PA approach.

Classic PA algorithms include $2$D \cite{ChowdhuryRa08, Cannon69}, 
$3$D \cite{AgarwalBaGu95}, or $2.5$D \cite{SolomonikDe11}. 
These algorithms assume a square MM and require that
processor number $p$ be factorizable into two or three
roughly equal numbers.
Aggarwal et al. \cite{AggarwalChSn90} proved a lower bound
as well as a matching $3$D square MM algorithm on their 
shared-memory LPRAM model. Irony et al. \cite{IronyToTi04}
proved a lower bound for $2$D and $3$D square MM algorithms
on a distributed-memory model. 
McColl and Tiskin \cite{McCollTi99} provided a similar
$3$D square MM algorithm on their BSPRAM model.
Solomonik and Demmel 
\cite{SolomonikDe11} coined a $2.5$D square MM algorithm,
which can change its partitioning of computational DAG as 
well as processor grid according to the availability
of memory to achieve optimal communication complexity
on a distributed-memory model. 
Demmel et al. \cite{DemmelElFo13} proved the lower bound
as well as the first communication-optimal 
algorithm for all dimensions of 
rectangular MM. Their proof assumes that processor 
number $p$ is an exact power of $2$ and their algorithm 
is efficient by the proof if all prime factors of p 
can be bounded by a small constant. By contrast, our 
algorithm and proof work for an arbitrary number of 
processors, even when $p$ per se is a large prime number. 
Our algorithm matches the lower bound proved in
\cite{DemmelElFo13}. 
\subsubsecput{ext-distri}{Extension to a Distributed-Memory Computing System}
 
One of the reasons that we choose a PA approach is that PA
algorithms are more portable to both shared-memory and 
distributed-memory computing systems. Though Network-Oblivious
(NO) algorithms by Bilardi et al. \cite{BilardiPiPu16} and
Chowdhury et al. \cite{ChowdhuryRaSi13} are efficient on 
the D-BSP model, provided there is a provably efficient 
folding mechanism. Such a folding mechanism is not available
in practice. 
There are at least two ways to port a \ouralg{} 
to a distributed-memory computing system as follows. 
\begin{enumerate}
    \item If assuming that each processor has an arbitrarily 
        large local disk besides a local memory of size $Z$
        \footnote{
            This assumption can be valid by the virtual
            memory (VM) system (Chap. 9 of \cite{BryantOh15}). 
            A user's program can only access VM, which usually
            resides on a local disk. VM system will bring data
            to physical memory when user accesses it. A
            $32$-bit system usually has a $2^{32}$-byte VM, 
            while a $64$-bit system usually has a VM of size 
            $2^{64}$ or $2^{48}$ bytes, all of which are usually
            much larger than corresponding physical memory size.
        }
        , a \ouralg{}'s communication can be separated into two 
        phases. 

        The first phase will be an inter-processor message 
        passing, the bandwidth of which will be the 
        memory-independent communication bound proved for a
        \ouralg{}. In the case of \ourAlg{MM-1-piece}, the 
        latency bound will be $O(\log p)$. For each cuboid,
        the read of two side faces, i.e. sub-matrices of $A$ 
        and $B$ requires only $O(1)$ messages by a proper
        packing / unpacking. 
        The $O(\log p)$ latency comes from writing
        intermediate results back to $C$ because in the worst 
        case all cuts are on the height of initial cuboid, hence 
        requires $O(\log p)$ rounds for reduction.
        
        The second phase will be a local sequential computation,
        which will incur only sequential cache misses between 
        local memory / disk pair. The local bandwidth of
        this phase will be the memory-dependent or 
        memory-independent bound of each \ouralg{},
        depending on the relative size of surface area of cuboid 
        with respect to the local memory size $Z$.

    \item If assuming a distributed-memory model as in 
        \cite{DemmelElFo13}, i.e. each processor has only one
        local memory of size $Z$ with no local disk,
        then the bandwidth bound will still be the same as
        the communication bound proved for each \ouralg{}.
        Take the \ourAlg{MM-1-Piece}{} as an example, the latency
        bound will be a factor of $Z$ lower than the bandwidth
        bound as follows.
        The number of messages to compute a cuboid will be 
        $\min (V'/V_Z) \cdot \log p$, where $V'$ is the volume of 
        cuboid, which is $O(nmk/p)$, 
        and $V_Z$ is the largest volume of a cuboid that has 
        an $O(Z)$ surface area, which stands for the largest 
        amount of 
        multiplications that can be done by having $O(Z)$ elements.
        So $\min (V'/V_Z)$ accounts for the minimal number
        of messages for reading sub-matrices of $A$ and $B$, 
        and $\log p$ is for writing back to $C$.
        According to Loomis-Whitney Inequality 
        \cite{DemmelElFo13, LoomisWh49}, the largest volume 
        that a cuboid with surface area of $O(Z)$ can have is 
        $O(Z^{3/2})$, i.e. when the cuboid is a cube. 
        So the number of messages
        reduce to $O(nmk/(pZ^{3/2})\cdot\log p)$.
        The latency bound of CARMA \cite{DemmelElFo13} is
        different from this bound because they assume matrices
        $A$, $B$, and $C$ are stored distributedly among $p$ 
        processors' memory, hence not every
        intermediate results have to be written back to $C$.
\end{enumerate}

\subsubsecput{hetero-mm}{Extension to Heterogeneous Computing System}
 
The heterogeneous computing system considered in this section
makes following modifications to the \ourModel{}.
It has $p$ processors, each of which can have a different but
\emph{fixed} throughput.
In the case of MM, it means that if we execute the same-sized
MM sequentially on every computing cores, the throughput, 
say FLOPS (Floating Point Operations Per Second), of all
cores can be normalized to $t_1 : t_2 : \cdots : t_p$. 
For simplicity, we assume that this thoughput ratio is fixed
and does not change on different problem sizes.
Without loss of generality, we assume that the throughput ratio 
is in a monotonically non-decreasing order. 
That is, $t_1 = 1$, $\forall i, j \in [1, p]$, we have
$t_i \leq t_j$ if $i \leq j$,
where $t_i, t_j \geq 1$ are arbitrary real numbers.

We construct our \ourAlg{Hetero-MM}{} based on the $2$-way
divide-and-conquer procedure of \ourAlg{MM}{} as follows.
The intuition is to assign cuboids to processors proportional
to their throughput ratio, plus that all cuboids assigned
to any processor still keep a geometrically decreasing
sequence in terms of both volume and surface area.
Firstly, we normalize the throughput ratio 
to fraction ratio of $f_1 : f_2 : \cdots : f_p$,
where $f_i = t_i / \sum_{j=1}^{p} t_j$. Each fraction number
$f_i$, where $i \in [1, p]$, indicates the 
fraction of total computational loads to be assigned to 
processor-$i$. 
Secondly, we still perform a similar recursive $2$-way 
divide-and-conquer procedure to that of \ourAlg{MM}{}. 
In addition, we associate each cuboid with a real number to 
indicate its fraction of total computational loads. For 
example, the initial cuboid
of $n \times m \times k$ will have a fraction number $1$,
a cuboid of $n' \times m' \times k'$ will have a fraction 
number of $f' = \frac{n'm'k'}{nmk}$, and so on.
Thirdly, in the recursive $2$-way divide-and-conquer,
whenever a cuboid's fraction number $f'$ is less than or equal
some processor's remaining ratio $f_i$, we make an assignment
and adjust the processor's remaining ratio by $f'$, i.e.
$f_i = f_i - f'$. This recursive procedure
repeats until all remaining cuboids are of constant sizes, 
in which
case they will be assigned to all processors in a round-robin 
fashion.

By the modification, it's not hard to check that 
\punt{
the 
intuition holds because the invariant that each processor 
will not have more 
than one non-constant cuboid of the same size by the $2$-way 
divide-and-conquer still holds.
The intuition actually
suggests the following modification to the \obc{} and \occ{}
properties. That is, 
}
the amount of computation and communication 
assigned to every processors should be proportional to their 
throughput ratio, hence the running time on every processors
are identical. As a consequence, the algorithm will reach an
ideal speedup. 
\punt{
as shown in \corref{hetero-mm}.
}

\begin{corollary}
    The \ourAlg{Hetero-MM}{} multiplies an $n$-by-$k$ 
    matrix
    $A$ with an $k$-by-$m$ matrix $B$ on a heterogeneous computing
    system with $p$ processors of throughput ratio 
    $t_1 : t_2 : \cdots : t_p$, where $t_1 = 1$ and $t_i \geq t_1$
    for $1 < i \leq p$, in optimal $\tsum = O(nmk)$ work,
    with an $O(t^{\sum})$ speedup with respect to a sequential 
    execution on processor-$1$, 
    where $t^{\sum} = \sum_{j=1}^{p} t_j$.
\label{cor:hetero-mm}
\end{corollary}

\paragrf{Discussions: }
Our model for heterogeneous computing systems is simpler than
that in Ballard et al. \cite{BallardDeGe11}. Their model 
considers four parameters, i.e. $\beta_i$
(inverse bandwidth), $\alpha_i$ (latency), $M_i$ (local memory
size), and $\gamma_i$ (flops per second), for $1 \leq i \leq p$. 
We simplify it to just throughput ratio because we feel that 
the parameters $\alpha_i$, $\beta_i$, and $\gamma_i$ are closely
related in any real system and are usually proportional to each
other in an algorithm's complexity bound. They develop a 
heterogenous algorithm for square MM, and our
\ourAlg{Hetero-MM}{} works for a rectangular MM of all dimensions.
The same scheme extends to heterogeneous Strassen as well.

Beaumont et al. \cite{BeaumontBeDe19, BeaumontEyLa16} proposed
$2$D and $3$D approximate algorithms for partitioning square
MM on a heterogeneous computing system, with a proof that an
exact partitioning is NP-Complete. Their method is Non-Rectangular
Partitioning and has a better approximate ratio than the
Rectangular Partitioning proposed by Nagamochi and Abe 
\cite{NagamochiAb07}.

\subsecput{strassen}{\ourAlgS{Strassen}}

Assuming the existence of an inverse operation of addition,
Strassen's algorithm \cite{Strassen69} is a $2$-way
divide-and-conquer algorithm that recursively reduces
$1$ multiplication of two $n$-by-$n$ matrices to
$7$ multiplications of two $n/2$-by-$n/2$ matrices
plus a constant number of matrix additions and subtractions 
on a ring as follows.

\vspace*{-.5cm}
\begin{align*}
    C = 
    \begin{bmatrix}
        C_{00} & C_{01} \\
        C_{10} & C_{11}
    \end{bmatrix}
    ,
    A = 
    \begin{bmatrix}
        A_{00} & A_{01} \\
        A_{10} & A_{11}
    \end{bmatrix} 
    ,
    B = 
    \begin{bmatrix}
        B_{00} & B_{01} \\
        B_{10} & B_{11}
    \end{bmatrix} 
\end{align*}
%
\begin{tabular}{lll}
$S_1 = A_{00} \oplus A_{11}$ & $S_2 = A_{10} \oplus A_{11}$ & $S_3 = A_{00}$ \\
$S_4 = A_{11}$ & $S_5 = A_{00} \oplus A_{01}$ & $S_6 = A_{10} \ominus A_{00}$ \\
$S_7 = A_{01} \ominus A_{11}$ & $T_1 = B_{00} \oplus B_{11}$ & $T_2 = B_{00}$ \\
$T_3 = B_{01} \ominus B_{11}$ & $T_4 = B_{10} \ominus B_{00}$ & $T_5 = B_{11}$ \\
$T_6 = B_{00} \oplus B_{01}$ & $T_7 = B_{10} \oplus B_{11}$ &  
\end{tabular}
\begin{tabular}{cc}
$M_r = S_r \otimes T_r$, & $1 \leq r \leq 7$
\end{tabular}
\hspace*{-.25cm}
\begin{tabular}{ll}
$C_{00} = M_1 \oplus M_4 \ominus M_5 \oplus M_7$ & $C_{01} = M_3 \oplus M_5$ \\
$C_{10} = M_2 \oplus M_4$ & $C_{11} = M_1 \oplus M_3 \ominus M_2 \oplus M_6$
\end{tabular}
\punt{
\begin{align}
    C_{00} &= M_0 \oplus M_3 \ominus M_4 \oplus M_6\nonumber\\
    C_{01} &= M_2 \oplus M_4\nonumber\\
    C_{10} &= M_1 \oplus M_3\nonumber\\
    C_{11} &= M_0 \ominus M_1 \oplus M_2 \oplus M_5 \label{eq:strassen-merge}
\end{align}

\setlength{\columnsep}{0pt}
\begin{multicols}{3}
\begin{align}
    M_0 &= S_0 \otimes T_0\nonumber\\
    M_1 &= S_1 \otimes T_1\nonumber\\
    M_2 &= S_2 \otimes T_2\nonumber\\
    M_3 &= S_3 \otimes T_3\nonumber\\
    M_4 &= S_4 \otimes T_4\nonumber\\
    M_5 &= S_5 \otimes T_5\nonumber\\
    M_6 &= S_6 \otimes T_6\label{eq:strassen-times-M}
\end{align}
\vfill\null
\columnbreak
\begin{align}
    S_0 &= A_{00} \oplus A_{11}\nonumber\\
    S_1 &= A_{10} \oplus A_{11}\nonumber\\
    S_2 &= A_{00}\nonumber\\
    S_3 &= A_{11}\nonumber\\
    S_4 &= A_{00} \oplus A_{01}\nonumber\\
    S_5 &= A_{10} \ominus A_{00}\nonumber\\
    S_6 &= A_{01} \ominus A_{11}\label{eq:strassen-additions-A}
\end{align}
\vfill\null
\columnbreak
\begin{align}
    T_0 &= B_{00} \oplus B_{11}\nonumber\\
    T_1 &= B_{00}\nonumber\\
    T_2 &= B_{01} \ominus B_{11}\nonumber\\
    T_3 &= B_{10} \ominus B_{00}\nonumber\\
    T_4 &= B_{11}\nonumber\\
    T_5 &= B_{00} \oplus B_{01}\nonumber\\
    T_6 &= B_{10} \oplus B_{11}\label{eq:strassen-additions-B}
\end{align}
\end{multicols}
}

We can view the computation of Strassen's algorithm as a 
cube of dimensions $n$-by-$n$-by-$n$, where the two side faces
stand for the input matrices $A$ and $B$, and the bottom face
stands for the output matrix $C$, respectively.
Our \ourAlgS{Strassen}{} is then a pruned BFS traversal of 
a $7$-way divide-and-conquer tree as follows. 
Each node of tree stands for a matrix multiplication, which is
also called a cube in our description, and
the seven children nodes of it are the seven ($7$) derived 
smaller-scale cubes.
\punt{
The overheads of constant number of matrix additions and 
subtractions of each node are charged to corresponding 
multiplications.
}
All intermediate matrices, i.e. $S$, $T$, and $M$, are held
in temporary space so that all derived nodes of the
same depth can run concurrently. 
As soon as some depth contains equal or more than $p$
unassigned nodes, exact $p$ of them will be pruned and assigned 
to $p$ processors in a round-robin fashion. The rest of
nodes, if any, will go to the next round of division. This 
procedure repeats until all nodes on the same depth are of 
base (constant) sizes, in which case all of them will be 
pruned and assigned in a round-robin fashion.
An assigned node stops any further parallel 
divide-and-conquer and will be executed by the cache-oblivious
sequential Strassen's algorithm \cite{FrigoLePr12} on the 
assigned processor.
The entire procedure is similiar to that shown in
\figref{pruned-bfs-pic}, except that it is now a a $7$-ry tree.

\begin{theorem}
    The \ourAlgS{Strassen}{} multiplies two
    $n$-by-$n$ matrices in optimal $\tsum = O(n^{\omg})$ work,
    optimal $\tinf = O(n^{\omg}/p)$ time, 
    using $O(p^{0.29} n^2)$ temporary space,
    $\qsum = O(n^{\omg}/(LZ^{\omg / 2 - 1}) + n^2/(Lp^{2/\omg - 1}))$, and  
    $\qinf = (1/p) \qsum$, 
    where $\omg = \log_2 7$ and $0.29 \approx 1 - \log_7 4$,
    assuming $p = o(n)$.
    The \perfectspeedup{} range is $n = \Omega(Z)$. 
\label{thm:strassen}
\end{theorem}
\begin{proof}
    The conclusion of \obc{} is clear from the algorithm, and the
    property of 
    \occ{} follows by showing that the sequence of cubes, i.e.
    multiplications, assigned 
    to each processor forms an almost geometrically decreasing 
    sequence in terms of volume, i.e. $O(n^{\omg})$, and surface 
    area, i.e. $O(n^2)$, up to a constant factor of $6$, and
    that the cache complexity of cache-oblivious sequential 
    Strassen's algorithm \cite{FrigoLePr12} is proportional to 
    the volume when
    the sizes of input and output, i.e. $3n^2$, is larger than
    the cache size $Z$, and proportional to the surface
    area otherwise.
    So the overheads of top-level nodes dominate on each processor.
    The overheads of constant number of matrix additions and
    subtractions of each node can be charged to corresponding
    multiplications.
    The temporary space before the first assignment of $p$ nodes
    is 
    $3 \cdot 7 \cdot n^2 \cdot \sum_{i = 0}^{\lceil \log_7 p \rceil} (7/2^2)^i = O(p^{\log_7 (7/4)} n^2) \approx O(p^{0.29} n^2)$,
    where $3 \cdot 7 \cdot n^2$ is the temporary space for 
    top-level recursion (to hold $S$, $T$, and $M$), and the 
    $\sum$ is to 
    accumulate over $\lceil \log_7 p \rceil$ recursion levels
    before the first assignment.
    Since it is pruned BFS traversal, later space requirement
    after the first assignment will be dominated.
    The \perfectspeedup{} range comes when the memory-dependent
    bound dominates.
\end{proof}

From \ourAlgS{Strassen}{}, we can see that
after first $i_1 = \lceil \log_7 p \rceil$ rounds of $7$-way
branching, each processor will be assigned up to $6$ same-sized
cubes, and will get the next assignment after another
$i_2 - i_1 = \lceil \log_7 (p / (7^{\lceil \log_7 p \rceil} 
- p)) \rceil$ rounds of $7$-way branching, and so on. 
If we denote the number of rounds that 
yields the $j$-th assignment by $i_j$, which we call 
$j$-th super-round, we make the following changes.
The new algorithm will stop parallel divide-and-conquer after 
$\sr$ super-rounds, where $\sr$ is some constant to be 
determined later. If there are still unassigned
cubes, the algorithm assigns all of them
to $p$ processors in a round-robin fashion.
Ignoring constant, the maximal possible difference
in computational loads among different 
processors is $f_{\text{comp}} = 1 - 
\frac{\sum_{j=1}^{\sr} (n/2^{i_j})^{\omg}}{\sum_{j=1}^{\sr} (n/2^{i_j})^{\omg} + (n/2^{i_\sr})^{\omg}} 
\leq 1 - \frac{(n/2^{i_1})^{\omg}}{(n/2^{i_1})^{\omg} 
+ (n/2^{i_\sr})^{\omg}} \leq 1 - \frac{2^{\sr - 1}}{2^{\sr - 1} + 1}$
The last inequality is because each super-round contains at
least one round of $7$-way branching.
We can see that $f_{\text{comp}}$ can be made arbitrarily
close to $0$ with the increase of $\sr$.
A similar conclusion applies to differnce in cache complexity 
as well.
Note that $\sr$ depends on processor number $p$, 
but is independent of problem size $n$.
These changes make our improved \ourAlg{Strassen-Const-Pieces}{}
(\corref{strassen-const-pieces}).
 
\begin{corollary}
    The \ourAlg{Strassen-Const-Pieces}{} 
    multiplies two $n$-by-$n$ matrices, by having only 
    constant pieces of cubes on each processor,
    in optimal $O(n^{\omg})$ work,
    optimal $O(n^{\omg}/p)$ time,
    using $O(p^{0.29} n^2)$ temporary space,
    $\qsum = O(n^{\omg}/(LZ^{\omg / 2 - 1}) + n^2/p^{2/\omg - 1})$, and 
    $\qinf = (1/p) \qsum$, 
    where $\omg = \log_2 7$ and $0.29 \approx 1 - \log_7 4$,
    assuming $p = o(n)$.
\label{cor:strassen-const-pieces}
\end{corollary}

In practice, we can make $\sr$ a tuning parameter. 
For example, if $\sr = 8$, the load imbalance among different
processors, if any, will be less than $1\%$.

\paragrf{Discussions: }\\
The load imbalance of \ourAlgS{Strassen}{} among different 
processors is an asymptotically smaller term, if any,
so is optimal in a shared-memory setting;
However, if translated to a distributed-memory setting, they may
have an $O(\log n)$ latency bound; By contrast, 
\ourAlg{Strassen-Const-Pieces}{} may
have an arbitrarily small constant-factor difference, 
but reduces latency to $O(\log p)$ in a distributed-memory 
setting.
The partitioning overheads of both our new Strassen's algoritms
can be fully parallelized and charged to each and every derived 
cubes as in the case of \ourAlg{1D}{} (see \figref{paco-1D-code}
for an analogue).
%
\paragrf{Open Problem on Parallelizing Strassen: }
Ballard et al. \cite{BallardDeHo12} developed a CAPS (
Communication-Avoiding Parallel Strassen) algorithm based
on interleaving of BFS/DFS steps on a distributed-memory model.
Their algorithm assumes that processor number $p$ is
an exact power of $7$. Lipshitz et al. \cite{LipshitzBaDe12} 
later improved it to a multiple of $7$ with no large prime
factors, i.e. $p = m \cdot 7^k$, where $1 \leq m < 7$ and 
$1 \leq k$ are integers, by a hybrid of Strassen and classic 
$O(n^3)$ MM algorithm.
They raised an open question in their paper 
(Sect. $6.5$ of \cite{BallardDeHo12}) whether a parallel
Strassen's algorithm can run on an arbitrary number of processors,
attains the computational lower bound exactly, and attains the 
communicational lower bound up to a constant factor.

If translated to a distributed-memory model, our
\ourAlg{Strassen-Const-Pieces}{} is an almost exact solution
to their open question, i.e. it runs concurrently on 
an arbitrary number of processors within a certain range,
attains
computational lower bound up to an arbitrarily small constant
factor, attains bandwidth lower bound up to a constant factor,
and attains the same $O(\log p)$ latency bound as the CAPS
algorithm. Moreover, our \ourAlg{Strassen-Const-Pieces}{} is
pure Strassen.
We further conjecture that this $O(\log p)$ latency bound is 
tight up to a constant factor.
Because in Strassen, each internal node of the $7$-ry tree
requires additional matrix additions and subtractions 
to construct new input matrices to the next level of recursion so
that an $\Omega(1)$ message(s) per node along a critical path 
seems inevitable. 
A parallel Strassen requires at least an $\Omega(\log p)$ depth 
to derive $\Omega(p)$ cubes of multiplications
\punt{\cite{JaJa92}}. So the $O(\log p)$ latency bound should be
tight up to a constant factor.

\paragrf{More Related Works on Parallel Strassen: }
McColl and Tiskin \cite{McCollTi99} developed a similar 
algorithm to the CAPS \cite{BallardDeHo12, LipshitzBaDe12} 
on their BSPRAM model.
McColl and Tiskin's algorithm is pure theoretical and ignores
certain practical considerations such as what if $p$ is not a
power of $7$.
Cole and Ramachandran \cite{ColeRa12a, ColeRa12} bounded the
overall parallel cache complexity of a
resource-oblivious Strassen, which belongs to the PO class,
to be $O(n^{\omg}/(LZ^{\omg/2 - 1}) + (p \log p)^{1/3} \cdot 
n^2/L + p \log p)$, which is asymptotically larger than all
PA (including PACO) counterparts.
Benson and Ballard \cite{BensonBa15} developed a code generation
tool to automatically implement multiple sequential and 
shared-memory parallel variants of fast MM algorithms.

\subsecput{sort}{\ourAlg{Sort}}

This section considers comparison-based sorting (sorting in 
short) algorithm. 

\begin{lemma} [\cite{FrigoLePr12}]
    There is a \proc{seq-sample-sort} algorithm that sorts
    $n$ elements by comparison in optimal $O(n \log n)$ work,
    and $O(1 + (n/L) (1 + \log_{Z} n))$ cache misses.
    \label{lem:seq-sample-sort}
\end{lemma}

Based on the sequential sameple sort \cite{FrigoLePr12} and
the observation that the maximal speedup a parallel algorithm
can attain on a $p$-processor system is $p$-fold if does not
count the caching effect, we have a \ourAlg{Sort}{} operating
on an array $A$ (stored in contiguous locations) of length $n$
as follows. We discuss the differences of our algorithms from
classic ones by the end of section.
\begin{enumerate}
    \item Picking $p-1$ pivots $\langle p_1, p_2, \cdots,
        p_{p-1} \rangle$ uniformly at random from the array 
        as follows.
        \begin{enumerate}
            \item Pick $kp$ samples uniformly at random from 
                the array, where $k$ is an over-sampling 
                ratio to be determined later. 
            \item Sort the $kp$ samples with the 
                \proc{seq-sample-sort} (\lemref{seq-sample-sort}).
            \item Pick every $k$-th sample as the final
                pivots.
        \end{enumerate}

    \item Redistributing elements of array $A$ by the 
        $p-1$ pivots as follows.
        \begin{enumerate}
            \item Each processor works simultaneously on
                a sub-array of length $n/p \pm 1$ of $A$ 
                and partitions it into $p$ partially ordered
                chunks by the $p-1$ pivots.
                That is, after the partitioning, all
                elements of the $i$-th chunk on any processor
                must be between
                the $(i-1)$-th and the $i$-th
                pivots in sorted order, $\forall i \in [1, p]$.
                The $0$-th and $p$-th pivots are defined
                to be $-\infty$ and $+\infty$ respectively, 
                
                This step can actually be performed by using
                a partial sequential quicksort \cite{Hoare62} 
                as follows. 
                Any processor-$i$ firstly partitions the 
                $i$-th sub-array by the 
                $\lceil p/2 \rceil$-th pivot 
                into two chunks such that all elements
                in the first chunk are less than or equal all
                elements in the second chunk. Then, each
                processor uses the
                $\lceil p/4 \rceil$-th and 
                $\lceil 3p/4 \rceil$-th pivots on the 
                first and second chunk, respectively,
                and so on for up to $\lceil \log_2 p \rceil$ 
                levels of recursion.
        
            \item Calculating the exact position of every
                chunk for re-distribution as follows.
                After the first step, we have a $p$-by-$p$
                matrix $[N]$, where each entry $n_{i, j}$ 
                stands for the number of elements of
                the $j$-th chunk on processor-$i$, which
                will be re-distributed to processor-$j$.
                By invoking a sequential prefix sum algorithm
                \punt{, say \ourAlg{Prefix Sum}{} }
                on every column of the matrix $[N]$ 
                simultaneously, we get each chunk's destined
                position for re-distribution.

            \item Performing a parallel matrix transposition 
                like the one in Blelloch et al. 
                \cite{BlellochGiSi10, BlellochLeMa98} to
                redistribute every chunk.
                That is, every processor
                will send $(p-1)$ chunks to other $(p-1)$ 
                processors by an all-to-all communication.
        \end{enumerate}

    \item Sorting locally, i.e. sequentially, on each 
        processor by the \proc{seq-sample-sort}
        (\lemref{seq-sample-sort}).
\end{enumerate}

\begin{theorem}
    The \ourAlg{Sort}{} sorts an array of $n$
    elements by comparison in optimal $\tsum = O(n \log n)$ work, 
    $\tinf = O((1 + \epsilon) n / p \cdot \log n)$ time for an 
    arbitrarily small $\epsilon \in (0, 1)$ with high probability, 
    using $O(p^2)$ temporary space,
    $\qsum = O((n/L) \log_{Z} (n/p))$, 
    and 
    $\qinf = (1/p) \qsum = O((n/pL) \log_{Z} (n/p))$,
    assuming $p \in O(\sqrt{n} / \ln n)$.
\label{thm:sort}
\end{theorem}
\begin{proof}
\paragrf{\Obc{}: }
By choosing an appropriate oversampling 
ratio $k > \frac{2 (c+1)}{(1+\epsilon)} \ln n$, where $c \geq 1$
and $0 < \epsilon < 1$ are some small constants, we can prove
that the number
of elements on each processor after re-distribution is no more 
than $(1 + \epsilon) n / p$ with probability $1 - n^{-c}$, 
i.e. with high probability.
The following proof adapts mostly from that of Theorem B.4. of 
\cite{BlellochLeMa98}. 
If we look at any particular element $i$
and its distance $I$ to the next pivot in sorted order. 
If $I \geq (1 + \epsilon) n/p$ elements, where $0 < \epsilon < 1$ 
is some small constant, there must be fewer than $k$ samples 
selected from these $(1 + \epsilon) n/p$ elements in sorted 
order. That is, $\func{Pr}[I \geq (1 + \epsilon) n/p] \leq 
\func{Pr}[Y_I < k]$, where \func{Pr} denotes the probability
of some event and $Y_I$ denotes the number of samples 
picked from these $(1 + \epsilon) n/p$ elements. Since the 
algorithm samples 
uniformly at random, each element has the same probability of 
$(kp/n)$ to be chosen. By the lower-tail Chernoff bound, 
$\func{Pr}[Y_I < k] = \func{Pr}[Y_I \leq (1 - \delta) \mu] \leq 
e^{-\mu \delta^2/2}$, where $\mu = (1 + \epsilon) k$ is the 
expected number of samples from $(1 + \epsilon) n/p$ elements 
and $\delta \in (\epsilon / (\epsilon + 1), 1)$ is a small 
variable to make the first equation of 
$\func{Pr}[Y_I < k] = \func{Pr}[Y_I \leq (1 - \delta) \mu]$ holds. 
Since this is the upper
bound for any single element to be within a balanced chunk. For
all elements to be within a balanced chunk, the probability is 
then no more than $n e^{-\mu \delta^2/2}$. To have a high 
probability bound, we make $n e^{-\mu \delta^2/2} \leq n^{-c}$, 
where $c \geq 1$ is some constant. Solving the inequality, we 
have the oversampling ratio of $k \geq 2 (c+1) \ln n / 
((1 + \epsilon) \delta^2) > \frac{2 (c+1)}{(1 + \epsilon)} \ln n$. 
Since $c \geq 1$ and $0 < \epsilon < 1$ are some small constants,
we conclude that $k \in O(\ln n)$.

The overall work of this algorithm sums up to the optimal
$O(kp \log kp) + O(n \log p) + O(p^2) + O(n) + O(n \log (n / p)) 
= O(n \log n)$, where $k = O(\log n)$ and assuming 
$p \in O(\sqrt{n})$. 
In the equation, $O(kp \log kp)$ is the work for sorting samples, 
$O(n \log p)$ is the work of using $(p-1)$ pivots to 
partition the array, $O(p^2)$ is the work of prefix sum on
$[N]$, $O(n)$ is the work for redistribution, and 
$O(n \log (n / p))$ is the overall work of final sequential 
sorting on every processors.

The time complexity is $O(kp \log kp) + O(n/p \log p) + O(p) + 
O((1 + \epsilon) n / p \cdot \log n) = 
O((1 + \epsilon) n/p \log n)$ for any arbitrarily small 
constant $0 < \epsilon < 1$ with high probability.
The $O(p^2)$ temporary space is used for storing matrix $[N]$
and computing the prefix sums. 

\paragrf{\Occ{}: }
\begin{enumerate}
    \item The parallel cache complexity of selecting pivots and 
        using the pivots to partition each sub-array of $n/p$ 
        elements into
        $p$ chunks is $O((kp / L) \log_{M} (kp) + (n / (pL)) 
        \log_{Z} (p))$ along the critical path and 
        $O((kp / L) \log_{M} (kp) + (n / L) \log_{Z} (p))$ 
        in summation.
        
    \item The parallel cache complexity of prefix sum and 
        redistribution is $O(n/(pL) + p)$ along the crticial 
        path and $O(n/L) + O(p)$ in summation.

    \item The parallel cache complexity of the final sequential
        sorting on each processor is 
        $O((n/pL) \log_{Z} (n/p))$ along the critical
        path and $O((n/L) \log_{M} (n/p))$ in summation.
\end{enumerate}
Summing up over all above overheads, we have the final parallel
cache complexity of 
$\qinf = O((n/(pL)) \log_{Z} (n/p))$ along the critical 
path and $\qsum = O((n/L) \log_{Z} (n/p))$ in summation
, assuming $p \in O(\sqrt{n}/\ln n)$.
\end{proof}

Note that the overall parallel cache complexity ($\qsum$) of 
\ourAlg{Sort}{} is actually smaller than the best sequential
cache bound of \proc{seq-sample-sort} (\lemref{seq-sample-sort}) 
because we have $p$ caches in the parallel setting and all 
procedures of sampling, partitioning and sequential sorting after 
re-distribution are concurrent on $p$ caches.

\paragrf{Discussions: }
Our algorithm is a variant of parallel sample sorting
algorithm. Parallel sample sorting algorithm has been 
studied in both PA \cite{BlellochLeMa98}
and PO fashions \cite{BlellochGiSi10}.
Cole and Ramachandran \cite{ColeRa17} developed a 
resource-oblivious, which also belongs to the PO class, sorting
which interleaves the partitioning of a sample sort with merging,
hence has only an $O(\log n \log \log n)$ critical-path length.

There are several key differences of our algorithm from 
the PO algorithm in \cite{BlellochGiSi10}.
Firstly, we use $(p-1)$ pivots instead of $O(\sqrt{n})$ pivots. 
Secondly, we call the sequential sample sort 
(\lemref{seq-sample-sort}) to sort 
on each processor after re-distribution, rather than 
a recursive low-depth one.
As a consequence, all PO algorithms \cite{ColeRa17, BlellochGiSi10}
incur more cache misses than the best sequential cache bound, 
though they all have a 
poly-logarithmic, i.e. low-depth critical-path length.
By contrast, our \ourAlg{Sort} incurs less.
As we can see from the experimental data in \secref{expr}, our 
algorithm does outperform the PO counterpart implemented
in PBBS \cite{ShunBlFi12} significantly.

The main difference of our algorithm from the PA
version in \cite{BlellochLeMa98} is that the early 
distributed-memory version calls a standard sequential radix sort 
after re-distribution for an empirical efficiency, while we call 
the sequential sample sort (\lemref{seq-sample-sort}) for an 
emphasis on \occ{}.
Putze et al. \cite{PutzeSaSi07}'s MCSTL utilizes atomic operations
for an in-place parallel quicksort with dynamic load-balance.

\secput{expr}{Preliminary Experimental Results}

We implement our \ouralgs{} and compare them 
on a $72$-core machine and a $24$-core machine
(\tabref{machine-spec}). 

\begin{table}[!ht]
\caption{Experimental Machines}
\label{tab:machine-spec}
\begin{tabular}{c|cc}
\toprule
Name & $72$-core machine & $24$-core machine \\
\midrule
OS   & CentOS 7.1 x86\_64 & CentOS 7 x86\_64\\
\hdashline
Compiler & ICC 15.0.2 & ICC 19.0.3 \\
\hdashline
\multirow{2}{*}{CPU}  & Intel Xeon & Intel Xeon\\
         & E7-8890 v3 & E5-2670 v3\\
\hdashline
Clock Freq & 2.50 GHz & 2.30 GHz\\
\hdashline
\# sockets & 4 & 2 \\
\hdashline
\# cores / socket & 18 & 12\\
\hdashline
Dual Precision & \multirow{2}{*}{16} & \multirow{2}{*}{16} \\
FLOPs / cycle & & \\
\hdashline
L1 dcache / core & 32 KB & 32 KB \\
\hdashline
L2 cache / core & 256 KB & 256 KB \\
\hdashline
L3 cache (shared) & 45 MB & 30 MB \\
\hdashline
memory & 128 GB & 132 GB \\
\bottomrule
\end{tabular}
\end{table}

\paragrf{Overview of Performance Comparison: }
Since the focus of this paper is to provide a new partitioning
and scheduling method of cache-oblivious algorithm, we request
that all algorithms of the same problem call the same 
kernel function(s) to compute sequentially base cases. 
For example, when comparing \ourAlg{MM}{} with
Intel MKL or PO counterpart,
we call MKL's sequential \func{dgemm} and
\func{daxpy} subroutines for base-case matrix
multiplications and additions, respectively
\footnote{Intel MKL actually does not have any
    subroutine for matrix addition and \func{daxpy} is
    for vector addition. So we call \func{daxpy} multiple
    times for our purpose.
}.
By this way, the only difference between peer 
algorithms is how they partition and schedule tasks.
We include all partitioning and scheduling overheads in final 
running time.
To avoid averaging noise, we measure ``running time'' as 
a $\min$ of at least three independent runs.
We calculate speedup by ``$(\text{running\_time}_{\text{peer alg.}} 
/ \text{running\_time}_{\text{\algName}} - 1) \times 100\%$''.

\subsecput{expr-mm}{MM}
\begin{table}[!h]
\centering
\begin{tabular}{cccc}
\toprule
$\id{Rmax} / \id{Rpeak}$ & PACO & MKL & CO2 \\
\midrule
Mean & $82.6\%$ & $75.1\%$ & $37.8\%$ \\
Median & $84.0\%$ & $78.4\%$ & $39.3\%$ \\
\bottomrule
\end{tabular}
\caption{$\id{Rmax}/\id{Rpeak}$ of MM algorithms. 
``CO2'' is the PO depth-$n$ MM algorithm based on $2$-way
divide-and-conquer with a base-case size of $64$
\cite{FrigoSt09}. ``Mean'' and ``Median'' is the mean and
median $\id{Rmax} / \id{Rpeak}$ of all data.}
\label{tab:24c-mm-rmax-rpeak}
\end{table}

\begin{figure*}[!ht]
    \begin{subfigure}[b]{0.47 \linewidth}
    \centering
    \includegraphics[width = \textwidth]{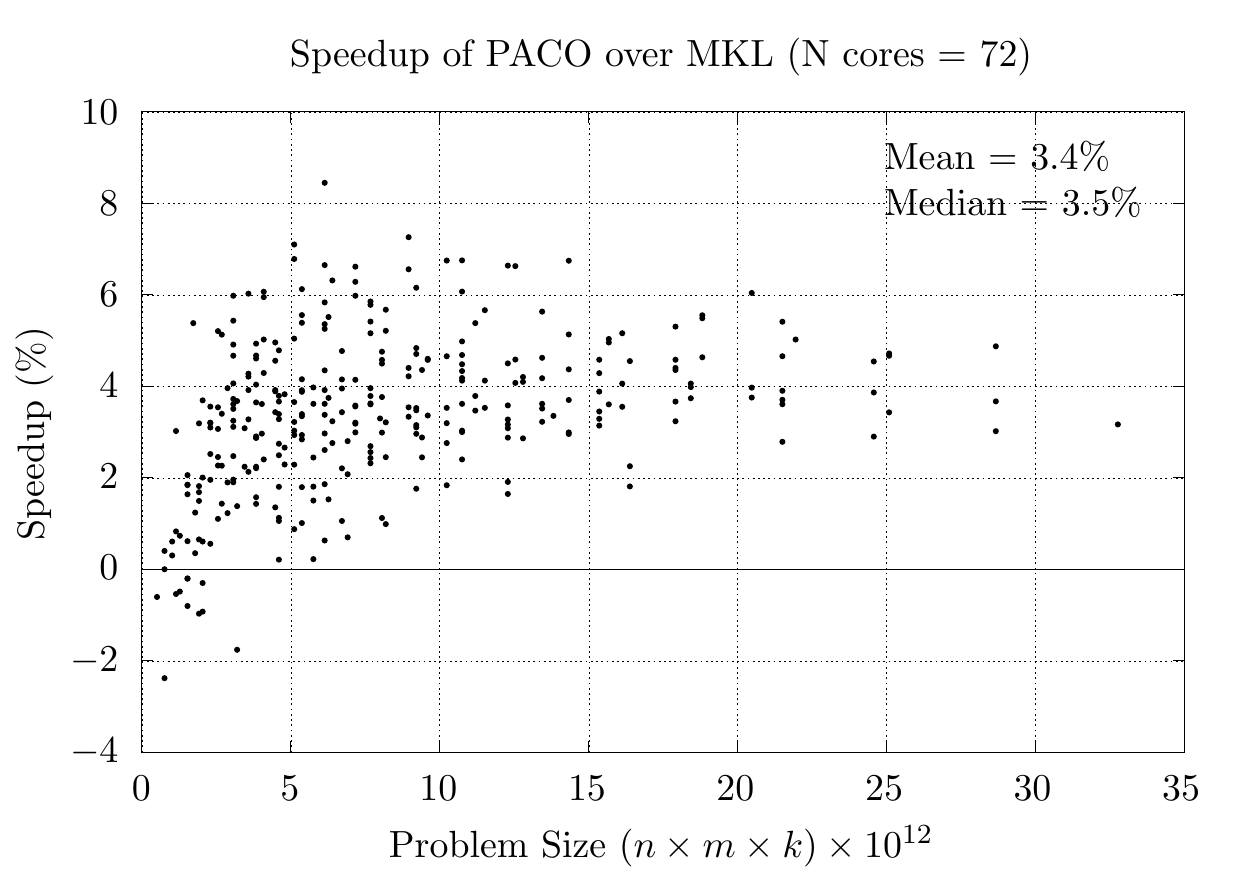}
    \caption{\ourAlg{MM-1-Piece}'s speedup over MKL's \func{dgemm}}
    \label{fig:72c-mm-distri-paco-mkl}
    \end{subfigure}
    \begin{subfigure}[b]{0.47 \linewidth}
    \centering
    \includegraphics[width = \textwidth]{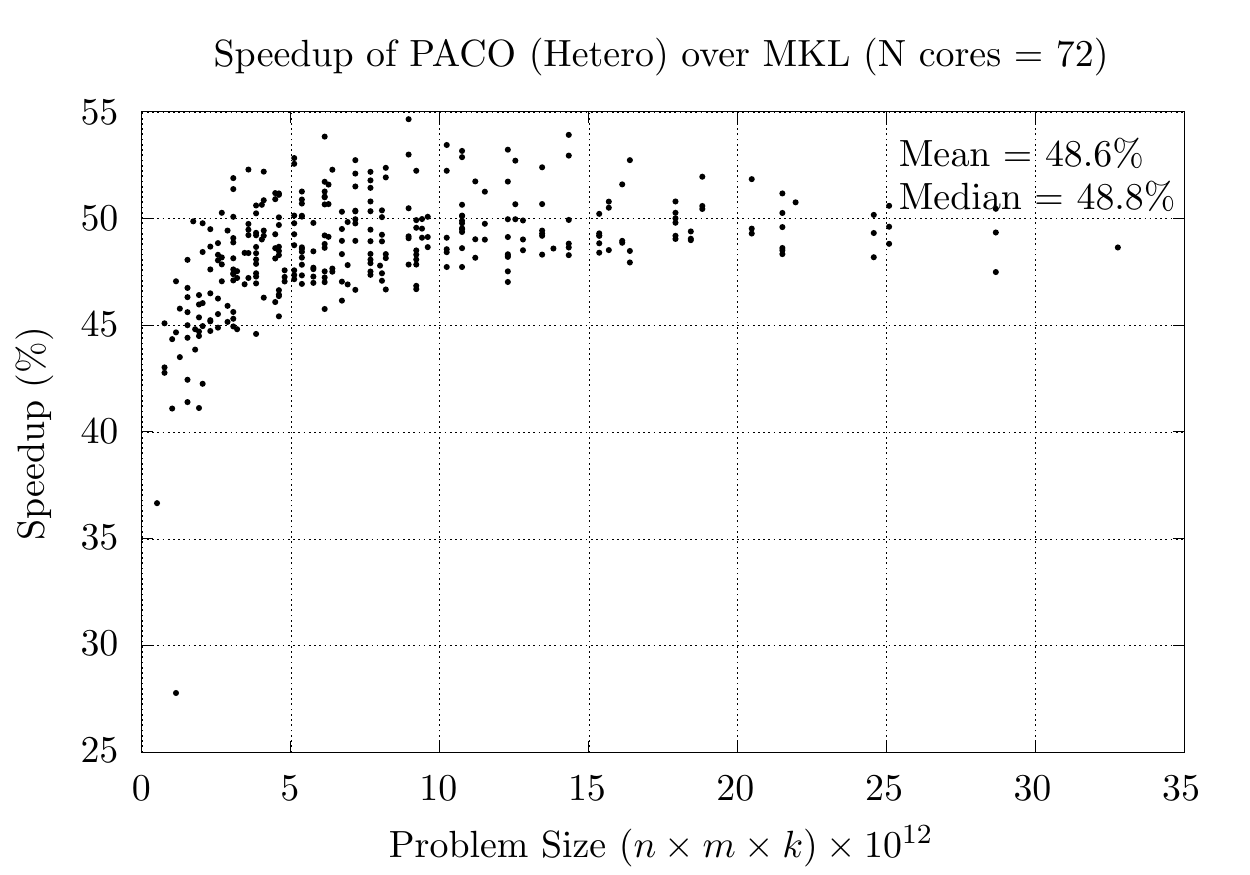}
    \caption{\ourAlg{Hetero-MM}'s speedup over MKL's \func{dgemm}}
    \label{fig:72c-mm-distri-paco-h-mkl}
    \end{subfigure}
    \caption{Speedup of \ourAlg{MM-1-Piece}{} over Intel MKL's \func{dgemm} on $72$-core machine.}
    \label{fig:72c-mm-perf}
\end{figure*}

\begin{figure*}[!ht]
    \begin{subfigure}[b]{0.47 \linewidth}
    \centering
    \includegraphics[width = \textwidth]{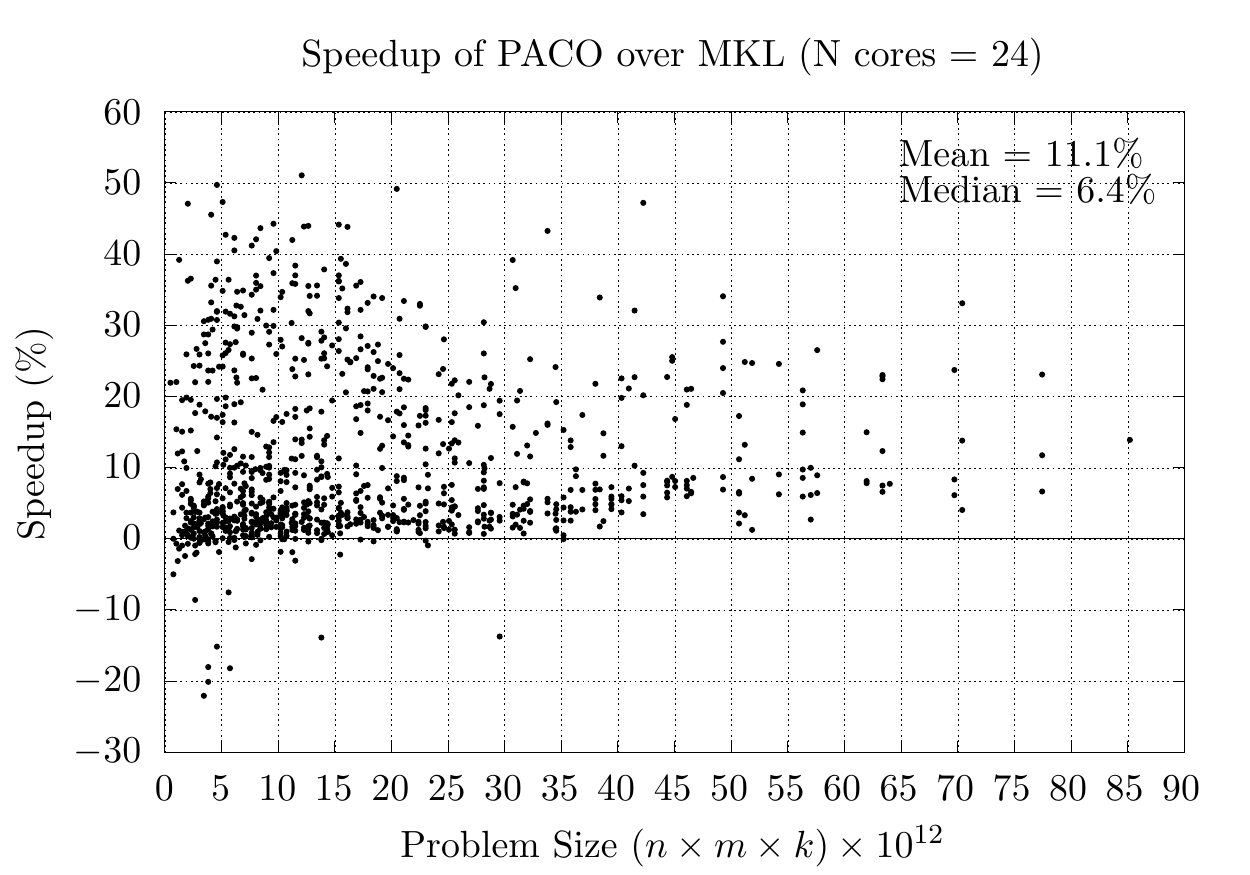}
    \caption{\ourAlg{MM-1-Piece}'s speedup over Intel MKL's \func{dgemm}}
    \label{fig:24c-mm-distri-paco-mkl}
    \end{subfigure}
    \begin{subfigure}[b]{0.47 \linewidth}
    \centering
    \includegraphics[width = \textwidth]{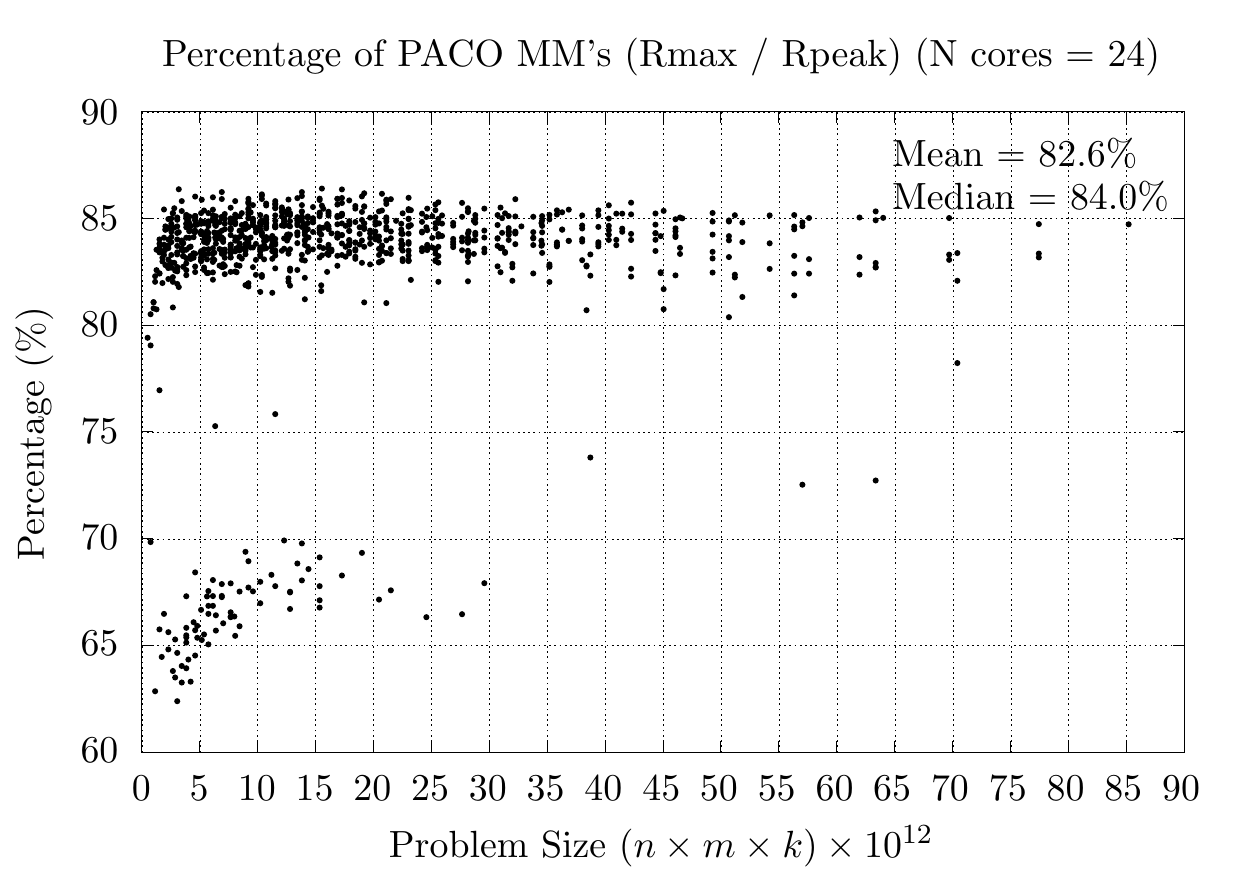}
    \caption{\ourAlg{MM-1-Piece}'s $\id{Rmax}/\id{Rpeak}$}
    \label{fig:24c-mm-flops-paco}
    \end{subfigure}
    \caption{Performance of \ourAlg{MM-1-Piece}{} on $24$-core machine.}
    \label{fig:24c-mm-paco-mkl}
\end{figure*}

\begin{figure*}[!ht]
    \begin{subfigure}[b]{0.47 \linewidth}
    \centering
    \includegraphics[width = \textwidth]{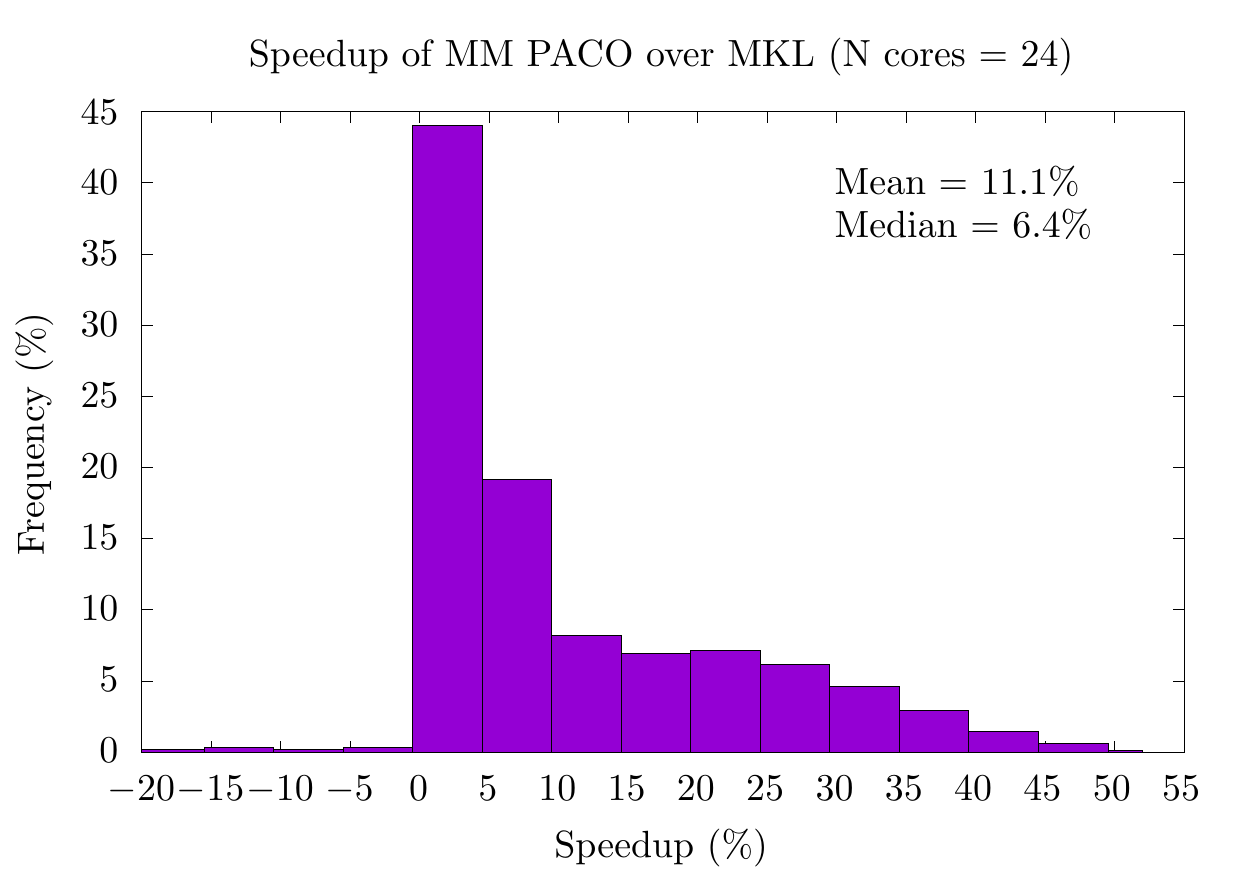}
    \caption{\ourAlg{MM-1-Piece}'s speedup over Intel MKL's \func{dgemm}}
    \label{fig:24c-mm-freq-paco-mkl}
    \end{subfigure}
    \begin{subfigure}[b]{0.47 \linewidth}
    \centering
    \includegraphics[width = \textwidth]{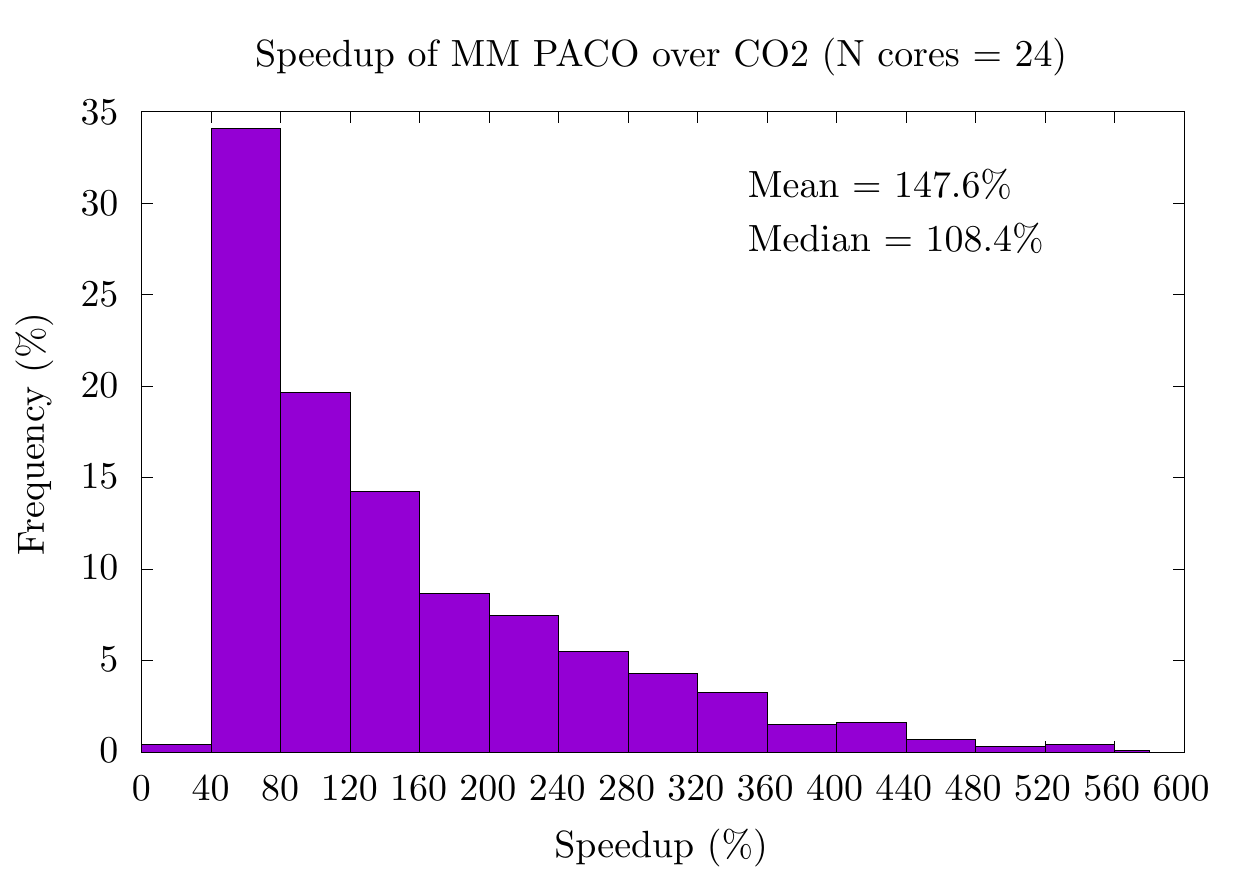}
    \caption{\ourAlg{MM-1-Piece}'s speedup over PO counterpart}
    \label{fig:24c-mm-freq-paco-co2}
    \end{subfigure}
    \caption{Performance of \ourAlg{MM-1-Piece}{} on the 
        $24$-core machine.
        ``CO2'' is the PO depth-$n$ MM algorithm based on $2$-way 
        divide-and-conquer with a base-case size of $64$ 
        \cite{FrigoSt09}.}
    \label{fig:24c-mm-more-perf}
\end{figure*}

We firstly compare \ourAlg{MM-1-Piece}{} (\corref{mm-1-piece})
with Intel MKL's parallel \func{dgemm} on the $72$-core machine.
\afigref{72c-mm-distri-paco-mkl} shows speedup distribution
along problem sizes. Problem size is calculated as
$n \times m \times k$ for an $n$-by-$k$ matrix multiplying
an $k$-by-$m$ matrix, where $n, m, k$ iterate independently from
$8,000$ to $44,000$ with a step size $4,000$.
So there are multiple points of the same $x$-value.
From the figures, though \ourAlg{MM-1-Piece}{} outperforms MKL in
majority of cases, the mean and median of speedup
is just $3.4\%$ and $3.5\%$, respectively.
\afigref{machine-spec} shows that this $72$-core machine has
$4$ sockets, each of which has $18$ cores.
Profiling shows that the $18$ cores on $0$-th 
socket are actually $3$ times faster than the other $54$ cores
on other $3$ sockets, though all these cores have the same clock
frequency and cache parameters.
By $3$ times faster, we mean that the $18$ cores on $0$-th socket
takes only $1/3$ of time of other $54$ cores when we compute
same-sized MM sequentially on every core.
Because the focus of this paper is algorithm, rather than 
systems or computer architecture, instead of figuring out
the reason of machine's heterogeneity, we simply switch to 
a heterogeneous version with the new results 
shown in \figref{72c-mm-distri-paco-h-mkl}.
Now we can see that the mean and median of speedup ratio
raises to $48.6\%$ and $48.8\%$, respectively.
To reduce the overheads of reduction of intermediate results, our
heterogeneous MM is slightly different from \ourAlg{Hetero-MM}{}
in \secref{ext-hetero}, but is similar to the
rectangular partitioning by Nagamochi and Abe \cite{NagamochiAb07}.
The algorithm structure is similar to \ourAlg{MM-1-Piece}{}
and has following changes.
We view the recursive divide-and-assign procedure as a binary
tree, where each leaf stands for a processor's throughput and
each internal node stands for the summation of its left and
right child's throughput.
We divide the initial cuboid starting from the root of tree
by a recursive procedure until each derived cuboid reaches a
leaf.
At each internal node, we cut a cuboid on its longest dimension
by the ratio of the node's left and right child's throughput.
By the change, each processor will get only one cuboid
instead of a sequence.

\afigref{24c-mm-distri-paco-mkl} shows the performance comparison 
of \ourAlg{MM-1-Piece}{} with MKL's parallel \func{dgemm} on the 
$24$-core system, with a mean speedup of $11.1\%$ and 
median of $6.4\%$.
Problem size is calculated as
$n \times m \times k$ for an $n$-by-$k$ matrix multiplying
an $k$-by-$m$ matrix, where $n, m, k$ iterate independently from
$8,000$ to $44,000$ with a step size of $4,000$.
So there are multiple points of the same $x$-value.
\afigref{24c-mm-more-perf} show the frequencies of 
\ourAlg{MM-1-Piece}{}'s speedup over MKL and PO counterparts. 
``CO2'' in the figure stands for the 
PO depth-$n$ MM algorithm based on $2$-way divide-and-conquer
\cite{FrigoSt09, BlellochChGi08} with a base-case size of $64$
\footnote{A $64$ base-case size means that the algorithm 
stops cutting a dimension when it is less than or equal
$64$ and a cuboid will be a base case when none of its three 
dimensions can be cut.}.
We select this base-case size by several manual trials to give
the CO2 algorithm a reasonably good performance on the machine,
though we do not attempt to make a thorough searching because
tuning is not the focus of this paper. Recent research by 
Leiserson et al. \cite{LeisersonThEm20} actually justifies 
our conclusion by showing
that a well-tuned PO MM algorithm achieves about $40\%$ of 
machine's peak performance.
Actually one concern on the PO approach is that its 
implementation may require to choose a proper base-case size,
i.e. when to stop
partitioning and parallelizing the algorithm, to balance
communication, synchronization and processor utilization.
If a base-case size is too small,
it increases ``slackness'' of algorithm and allows better
processor utilization for a wider range of processor counts,
but at the cost of more deviations from its sequential execution
order \cite{AcarBlBl00, SpoonhowerBlGi09}, hence more
communication and synchronization overheads.
On the other hand, if a
base-case size is too large, a base-case task may not fit in
some upper-level cache(s) of each processor, hence it may
not be cache-efficient, and the load imbalance among
processors may be larger, in other words, some processor may
be under-utilized.
By contrast, our approach does not need to tune.
\afigref{24c-mm-flops-paco} shows the percentages of
theoretical peak performance ($\id{Rmax} / \id{Rpeak}$) 
that \ourAlg{MM-1-Piece}{} has attained. 
\tabref{24c-mm-rmax-rpeak} lists different 
algorithm's mean and median of $\id{Rmax}/\id{Rpeak}$ 
side-by-side.
%
The \id{Rmax} is calculated by 
``$2 \times n \times m \times k / \id{time\_in\_second}$'' 
because we are computing $C = C + A \times B$ so there 
are $nmk$ multiplications and $nmk$ additions.
The \id{Rpeak} is calculated by ``$24 \times (2.3 \cdot 10^9) 
\times 16$'' because this machine has $24$
cores, each of which is $2.3$ GHz, which means $2.3 \cdot 10^9$
cycles per second, and each core can perform $16$ dual precision
floating point operations \footnote{By Fused Multiply Add 
\func{FMA3} instruction} per cycle.

\subsecput{expr-lcs-sort}{LCS and Sorting}

\begin{figure*}[!ht]
    \begin{subfigure}[b]{0.47 \linewidth}
    \centering
    \includegraphics[width = \textwidth]{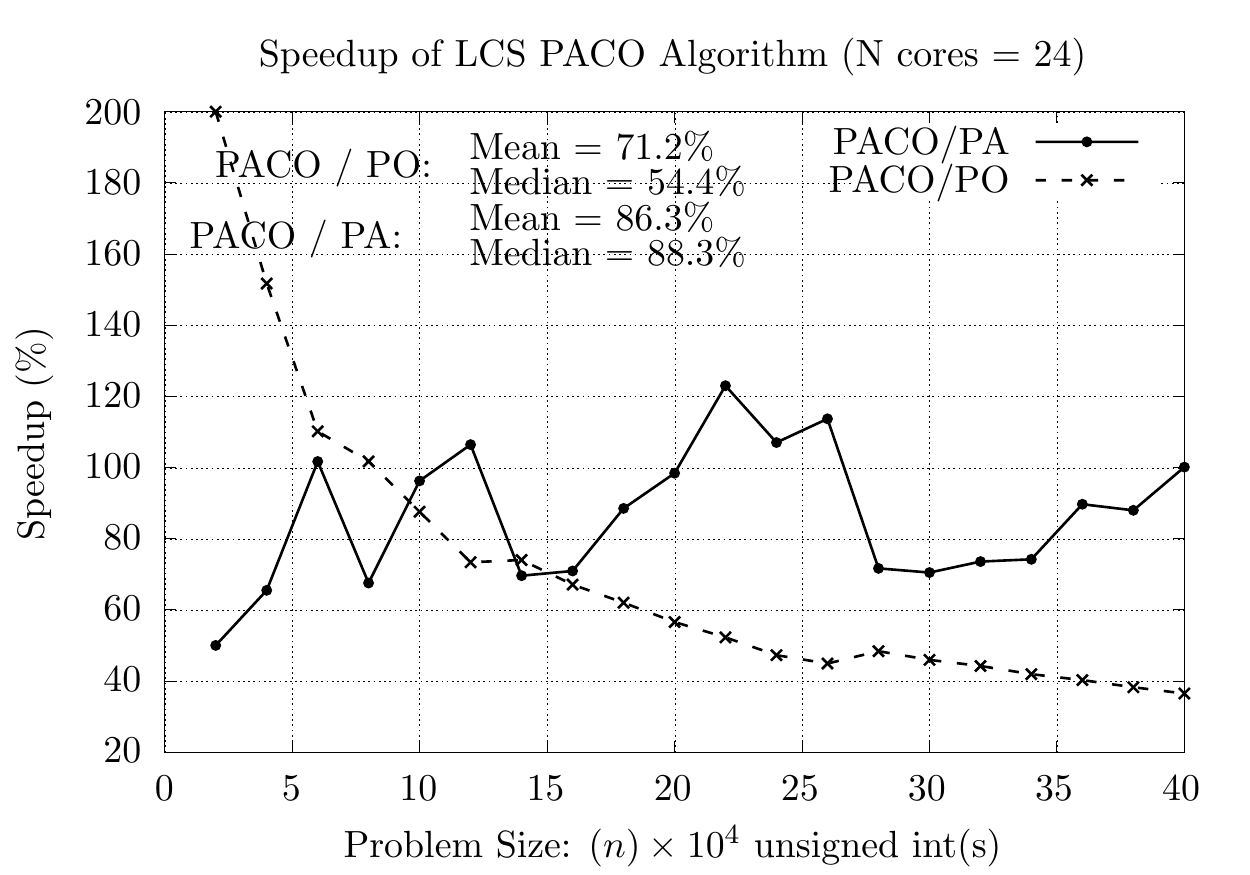}
    \caption{\ourAlg{LCS}'s speedup over PO and PA counterparts.}
    \label{fig:24c-lcs-distri}
    \end{subfigure}
    \begin{subfigure}[b]{0.47 \linewidth}
    \centering
    \includegraphics[width = \textwidth]{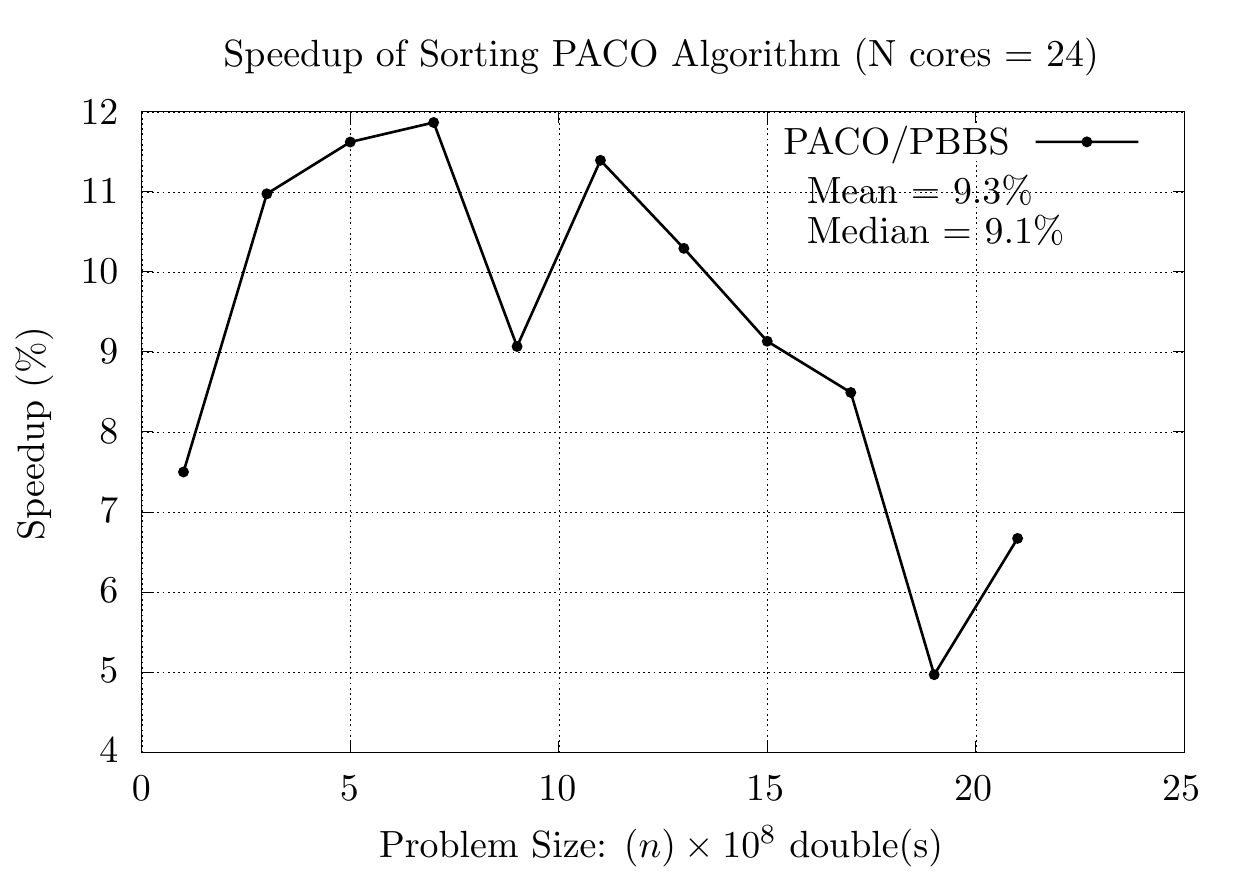}
    \caption{\ourAlg{Sort}'s speedup over PBBS \cite{ShunBlFi12}}
    \label{fig:24c-sort-distri}
    \end{subfigure}
    \caption{Experiments on LCS and Sorting algorithms on the $24$-core machine.}
    \label{fig:24c-lcs-sort}
\end{figure*}

We experiment \ourAlg{LCS}{} and \ourAlg{Sort}{} with 
PO and PA counterparts on
the $24$-core machine as shown in 
\figref{24c-lcs-sort}.
The PO LCS counterpart
is the classic $2$-way divide-and-conquer algorithm
\cite{CormenLeRi09} with a base-case size of $256$ elements
\footnote{A $256$ base-case size means that the algorithm
stops cutting a dimension when it is less than or equal
$16 = \sqrt{256}$ and a square will be a base case when none 
of its two dimensions can be cut.}, while the PA is the $p$-way 
divide-and-conquer by Chowdhury and
Ramachandran \cite{ChowdhuryRa08}.  
We select this base-case size by several manual trials to give
the PO algorithm a reasonably good performance on the machine,
though we do not attempt to make a thorough searching because
tuning is not the focus of this paper.

The mean and median speedups of \ourAlg{LCS}{} over
the PO is $71.2\%$ and $54.4\%$, respectively, and over
the PA is $86.3\%$ and $88.3\%$, respectively.
The PO Sorting counterpart is the low-depth 
sorting algorithm \cite{BlellochGiSi10} implemented in 
the Problem Based Benchmark Suite (PBBS) \cite{ShunBlFi12}.
We directly use the default oversampling ratio and other 
parameters implemented in PBBS without any tuning.
The mean and median speedup of \ourAlg{Sort}{} over it
is $9.3\%$ and $9.1\%$, respectively.

\secput{concl}{Concluding Remarks}

\paragrf{More Related Works: }
Andreev and R\"{a}cke \cite{AndreevRa06} partitions a graph into
several equal-sized components while minimizing the capacity of
edges between different components. They did not consider 
minimizing computation and communication along a critical
path so that their solution may not be a \perfectspeedup{}
one.

\paragrf{Conclusions: }
This paper proposes a general \ouralg{} based on the observation
that the maximal speedup attainable on a $p$-processor system
is usually $p$-folds so that excessive parallelism may not be
necessary. Our methodology is to partition computation and 
communication evenly and recursively among $p$ processors
by a pruned BFS traversal of a cache-oblivious
algorithm's divide-and-conquer tree. Each processor
will have balanced computational and communicational overheads,
usually forming in a geometrically decreasing sequence.
We apply the idea
to several important cache-oblivious algorithms, including LCS,
which is Dynamic Programming (DP) with constant dependencies,
1D and GAP, both of which are DP with more-than-constant 
dependencies, classic rectangular MM on a semiring and
Strassen's algorithm, as well as comparison based sorting.
Compared with classic PA counterparts, our algorithms achieve
\perfectspeedup{} on an arbitrary number, even a prime number,
of processors within a certain range. Compared with classic PO
counterparts, our algorithms usually have better communication
complexities.
Our \ourAlg{Strassen-Const-Pieces}{} provides
an almost exact solution to the open question on parallelizing
Strassen's algorithm efficiently and exactly on an arbitrary
number of processors.
%
Our preliminary experimental results confirm the theoretical
predictions.
%
Our methodology may provide a new perspective on 
the fundamental open problem of extending a 
sequential cache-oblivious algorithm to an arbitrary architecture.
We leave an efficient parallelization of \ourAlg{LCS}{}'s 
partitioning phase to future research.

\bibliographystyle{IEEEtran}
\bibliography{papers}


\end{document}